\documentclass{article}

\usepackage[noadjust]{cite}
\usepackage{amsmath,amssymb,amsfonts,amsthm,bbm}
\usepackage{algorithmicx}
\usepackage{textcomp}
\usepackage{xcolor}
\def\BibTeX{{\rm B\kern-.05em{\sc i\kern-.025em b}\kern-.08em
    T\kern-.1667em\lower.7ex\hbox{E}\kern-.125emX}}
\usepackage[ansinew]{inputenc}          
\usepackage{psfrag}                     
\usepackage[dvips]{graphicx,gincltex,epsfig}            
\usepackage[all]{xy}
\usepackage{a4wide}
\usepackage[dvips]{epsfig}
\usepackage{enumerate}
\usepackage{xcolor}
\usepackage{tikz}
\usepackage{nccmath}
\usepackage{hhline} 
\makeatletter
\newcommand\HUGE{\@setfontsize\Huge{34}{60}}
\makeatother   
\usepackage{enumitem, hyperref}
\makeatletter
\def\namedlabel#1#2{\begingroup
    #2%
    \def\@currentlabel{#2}%
    \phantomsection\label{#1}\endgroup
}
\makeatother

\usepackage{algorithm,algpseudocode}

\theoremstyle{plain}
\newtheorem{theorem}{Theorem}

\newtheorem{lemma}[theorem]{Lemma}
\newtheorem{corollary}[theorem]{Corollary}
\newtheorem{proposition}[theorem]{Proposition}
\theoremstyle{definition}

\newtheorem{example}[theorem]{Example}
\newtheorem{remark}[theorem]{Remark}

\newcommand{\Z}{\mathbb{Z}}
\newcommand{\N}{\mathbb{N}}
\newcommand{\F}{\mathbb{Z}}

\newcommand{\nolla}{\mathbf{0}}

\newcommand{\bw}{\mathbf{w}}
\newcommand{\bd}{\mathbf{d}}

\newcommand{\bc}{\mathbf{c}}

\newcommand{\supp}{\textrm{supp}}

\newcommand{\bx}{\mathbf{x}}

\newcommand{\bz}{\mathbf{z}}

\newcommand{\y}{\mathbf{y}}
\newcommand{\by}{\mathbf{y}}

\newcommand{\LL}{\mathcal{L}}

\title{On Unique Error Patterns in the Levenshtein's Sequence Reconstruction Model
	\thanks{The authors were funded in part by the Academy of Finland grants 338797 and 358718. 		
	A shorter version of this article was presented in ISIT2023 \cite{junnila2023levenshtein}.}
}

\author{Ville Junnila, Tero Laihonen and Tuomo Lehtil{\"a}\thanks{The authors are  with the Department of Mathematics and Statistics,
University of Turku, Finland (e-mail: viljun@utu.fi, terolai@utu.fi, tualeh@utu.fi). Some of the work of the third author was performed at the Department of Computer Science, University of Helsinki, Finland.}%
}

\begin{document}

\maketitle

\begin{abstract}
In the Levenshtein's sequence reconstruction problem a codeword is transmitted through $N$ channels and in each channel a set of errors
is introduced to the transmitted word. 
In previous works, the restriction that each channel provides a unique output word has been essential. In this work, we assume only that each channel introduces a unique set of errors to the transmitted word and hence some output words can also be identical. As we will discuss, this interpretation is both natural and useful for deletion and insertion errors. We give properties, techniques and (optimal) results for this situation.   

Quaternary alphabets are relevant
due to applications related to DNA-memories. Hence, we introduce an efficient Las Vegas style
decoding algorithm for simultaneous insertion, deletion and substitution errors in $q$-ary
Hamming spaces for $q \geq 4$.

\end{abstract}

\noindent\textbf{Keywords:}
	Information Retrieval, DNA-memory, Levenshtein's Sequence Reconstruction, Decoding Algorithm, Substitution Errors, Deletion Errors, Insertion Errors.

\setcounter{page}{1} 
\section{Introduction}\label{SecIntro}

We study \textit{Levenshtein's sequence reconstruction problem} introduced in \cite{Levenshtein}. In particular, we consider insertion, deletion and substitution errors in  $q$-ary Hamming spaces. The topic has been widely studied during recent years \cite{yaakobi2018uncertainty, Uusi_Maria_Abu-Sini,   goyal2022sequence, gabrys2018sequence, Maria_Abu-Sini,  junnila2022levenshtein,  junnila2020levenshtein,  horovitz2018reconstruction, chrisnata2022correcting}. Levenshtein's original motivation came from molecular biology and chemistry, where  adding redundancy was not feasible. Recently, Levenshtein's problem has returned to the limelight with the rise of advanced memory storage technologies such as associative memories \cite{yaakobi2018uncertainty}, racetrack memories \cite{chee2018reconstruction} and, especially, DNA-memories \cite{Uusi_Maria_Abu-Sini}, where the information is stored to DNA-strands. In the information retrieval process from the DNA-memories multiple, possibly erroneous, strands are obtained, due to biotechnological limitations \cite{Uusi_Maria_Abu-Sini}, which makes Levenshtein's model suitable for this topic. Another interesting property which we obtain from DNA-applications, is the emphasis on  information based on a quaternary alphabet over the binary due to the four types of nucleotides in which the information is stored (see \cite{bornholt2016dna, church2012next, grass2015robust, yazdi2015dna, sabary2024survey} for information about DNA-memories).

We will denote the set $\{1,2,\dots,n\}$ by $[1,n]$ and by $\F_q^n$ the \textit{$q$-ary $n$-dimensional Hamming space}. For a word $\bw\in \F^q_n$, we use notation $\bw=w_1w_2\dots w_n$ where each $w_i\in [0,q-1]$. The  \emph{support} of a word $\bw=w_1\dots w_n\in \F_q^n$ is defined as $\supp(\bw)=\{i\mid w_i\neq 0\}$, the \textit{weight} of $\bw$ with $w(\bw)=|\supp(\bw)|$ and the \textit{Hamming distance} between $\bw$ and $\bz$ with $d(\bw,\bz)=w(\bw-\bz)$. For the \textit{Hamming balls} we use the notation $B_t(\bw)=\{\bz\in\F_q^n\mid d(\bw,\bz)\leq t\}$ and $|B_t(\bw)|=V_q(n,t)=\sum_{i=0}^t(q-1)^i\binom{n}{i}$. A \textit{code} $C$ is a nonempty subset of $\F^n_q$ and it has \textit{minimum distance} $d_{\min}(C)=\min_{\bc_1,\bc_2\in C,\bc_1\neq\bc_2}d(\bc_1,\bc_2)$. Furthermore, $C$ is an \textit{$e$-error-correcting} if $d_{\min}(C)\geq2e+1$. Moreover, we denote the \textit{zero-word} $00\cdots0\in\F^n_q$ by $\nolla$ or $0^n$. Finally, notation $a^j$ means $j$ consecutive symbols $a$ and we sometimes may concatenate these for a notation such as $0^i10^j$ which would be a binary word of length $i+1+j$ of weight one where the single symbol $1$ is the $(i+1)$th symbol. In a \textit{substitution error} a symbol in some coordinate position is substituted with another symbol, in an \textit{insertion error} a new symbol is inserted to the original word leading to a word of length $n+1$ and in a \textit{deletion error} a symbol is deleted from the original word leading to a word of length $n-1$. Each of these three types of errors is relevant for DNA-memories \cite{heckel2019characterization}.

For the rest of the paper, we assume the following: $C\subseteq \F^n_q$ is a  code, a \textit{transmitted word} $\bx\in C$ is sent through $N$ channels in which  insertion, deletion and substitution errors may occur and the number of each type of error is limited by some constant $t_i,$ $t_d$ or $t_s$, respectively. When the error type is clear from the context, we drop the subscript. In some cases we use indices for the notation of individual error types. In many previous works, it has been assumed that each channel gives a different output word. We refer to this model of the problem as a \emph{traditional} Levenshtein's channel model.
However,  in this paper,  
we usually assume instead that in each channel a different set of errors occurs to the transmitted word $\bx$ which is a natural  assumption as we will see in Section \ref{SecErrorPat}. 

In our model of Levenshtein's sequence reconstruction problem, a (multi)set of output words $Y$ is received through $N$ channels. Based on $Y$, we deduce the transmitted word $\bx$. However, this is sometimes impossible and we have to instead settle for a list of possible transmitted words $T(Y)$ such that $\bx\in T(Y)$. The maximum size of this list over all $\bx$ and $Y$ is denoted by $\LL$. 
The channel model is illustrated in Figure \ref{LevenshteinFig}.

\begin{figure}[tp]
	\centering
	\includegraphics[trim={0.5cm 22.77cm 0 0.81cm},clip,scale=0.40]{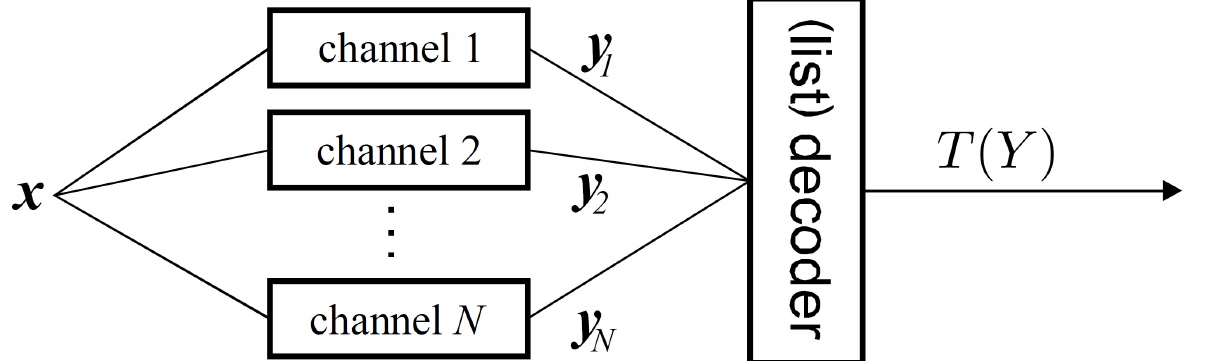}
	\centering\caption{The Levenshtein's sequence reconstruction.} \label{LevenshteinFig}
\end{figure}

If we have an $e$-error-correcting code $C$ and only substitution errors (and at most $t=e+\ell$) occur, then we have $$T(Y)=C\cap\bigcap_{\by\in Y} B_t(\by).$$ In this setup, the maximum size $\LL$ of $T(Y)$ has been studied in \cite{Levenshtein, yaakobi2018uncertainty, junnila2022levenshtein, junnila2020levenshtein} and it is well understood, for example, when $\LL$ is a constant. 
Especially, if $C = \Z_q^n$, $\LL = 1$ and exactly $t$ errors occur in a channel, then the required number of channels is given in the following theorem, which is based on Lemma~2 and Equation~(36) of \cite{levenshtein2001efficient}.

\begin{theorem}[\cite{levenshtein2001efficient}]\label{TheLevDelBound}
Let exactly $t\leq n-2$ deletion errors occur in the traditional Levenshtein's channel model and $C=\F^n_2$. Then $\LL=1$ if and only if $N\geq 2\sum_{i=0}^{t-1}\binom{n-t-1}{i}+1.$
\end{theorem}
The previous theorem gives the exact value for $N$ only when $q=2$. However, Levenshtein gave the number of channels also for cases with $q>2$. When $q=3$ an exact number is presented in \cite{levenshtein2001efficient} but when $q\geq4$, we know only a recursive formulation. Note that the case with $q=4$ has special relevance due to DNA related applications. In the following theorem, which is based on Equation (51) and Theorem 3 of \cite{levenshtein2001efficient}, the same question is discussed in the case of insertion errors.

\begin{theorem}[\cite{levenshtein2001efficient}]\label{TheLevInsBound}
Let exactly $t$ insertion errors occur in the traditional Levenshtein's channel model and $C=\F^n_q$. Then $\LL=1$ if and only if $N\geq \sum_{i=0}^{t-1}\binom{n+t}{i}(q-1)^i(1-(-1)^{t-i})+1.$
\end{theorem}

The relationship between the number of channels and the list size for deletion and insertion errors has  been recently studied in \cite{abu2023intersection}. Similarly to the article \cite{abu2023intersection}, we also mainly restrict in Section~\ref{SecErrorPat} our considerations in this paper to cases, where $C=\F^n_q$. Studying the problems discussed in this paper together with a code $C\neq \F^n_q$ would be interesting although more complicated. 

For combinations of different error types even less is known \cite{Uusi_Maria_Abu-Sini}. Moreover, when only deletion or insertion errors occur, there are open problems, for example, when the code $C$ or the list size $\LL>1$ \cite{gabrys2018sequence, Maria_Abu-Sini, Uusi_Maria_Abu-Sini, goyal2022sequence, levenshtein2001efficient}. Decoding algorithms for substitution errors have been studied in \cite{Levenshtein, yaakobi2018uncertainty, Uusi_Maria_Abu-Sini} and for deletion and insertion errors in \cite{gabrys2018sequence, levenshtein2001efficient, Uusi_Maria_Abu-Sini}.

Similar problem has also been considered, in some cases under the name \textit{trace reconstruction}, when each deletion/insertion/substitution has an independent probability to occur, that is, the maximum number of errors in a channel is not fixed unlike in the traditional model and our model. For example, with deletion errors, this would mean that each symbol has an independent probability of $\rho$ to be deleted in a channel. So, if only deletion errors occur, the chance to obtain the empty word is $\rho^n$; see, for example, \cite{batu2004reconstructing, cheraghchi2020coded, viswanathan2008improved, chen2022near}.

The structure of this paper is as follows: In Section \ref{SecErrorPat}, we consider the reconstruction problem when each channel has a unique error pattern. In particular, in Section \ref{sec:delVec}, we consider deletion error patterns. In Section \ref{sec:decoding} we continue by giving a fast online decoding algorithm when insertions, deletions and substitutions occur in channels. It requires only a linear time on the total length of (read) output words and does not always need to read every output word. It is likely that the algorithm returns a non-empty output, and the output of the algorithm is \emph{always correct} when it is non-empty. 


\section{Error patterns}\label{SecErrorPat}
Levenshtein's traditional channel model  has  usually been  considered under the assumption that every channel gives a unique output word \cite{levenshtein2001efficient, Uusi_Maria_Abu-Sini, junnila2020levenshtein}. In the case of deletion and insertion errors, this assumption of the model significantly restricts  the number of channels (depending on the transmitted word).
In this section, we consider different error types but our discussions often utilize deletion errors in the examples.
 Furthermore, for deletion errors and any $n,t\geq1$, there exist (as pointed out in \cite{levenshtein2001efficient}) cases in which we cannot deduce the transmitted word. This is also illustrated in the next example.
 
 \begin{example}\label{exDelTrad}
 	Let $\bx=111100$ and the number of deletions be \textit{exactly} $t=1$. If we require that each channel gives a unique output word, then we can have only two channels and $Y\subseteq B_1^d(\bx)\setminus\{\bx\}=\{11100,11110\}$ where $B_1^d(\bx)$ is the deletion ball of radius $1$, that is, the set of words which can be obtained from $\bx$ by removing at most one symbol. Moreover, if we consider a word $\bx'=111010$, then  $Y'\subseteq  B_1^d(\bx')\setminus\{\bx'\}=\{11010,11110,11100,11101\}$. In particular, $ B_1^d(\bx)\setminus\{\bx\}\subset  B_1^d(\bx')\setminus\{\bx'\}$. Consequently, if we only know the set $Y$, then we cannot say  in the traditional model with certainty whether the transmitted word is $\bx$ or $\bx'$! 
 \end{example}

Another challenge  for the traditional model  is, as we see above, that for deletion balls we may have $| B_1^d(\bx)\setminus\{\bx\}|\neq | B_1^d(\bx')\setminus\{\bx'\}|$.   In general, the size of the deletion ball depends on the choice of the central word \cite{levenshtein1966binary}.  In this section, we will introduce two new models (called a multiset model and a non-multiset model) which will help us with the problems of the traditional Levenshtein's model.  We also discuss why these models are natural for a wide range of parameters.

  \smallskip

Instead of assuming that every channel gives a unique output word, we will instead consider this problem with the assumption that each channel introduces a unique set of errors (a unique error pattern) to the transmitted word $\bx\in \F_q^n$. We mostly consider two types of error patterns. In the case of deletion errors, we introduce the concept of deletion vectors and for insertion errors, the concept of insertion vectors (introduced later in this section).
With a \textit{deletion vector} we mean a word $\bd\in \F_2^n$. When we apply the deletion vector $\bd$ to the transmitted word $\bx$, we obtain an output word $\by\in \F_q^{n-w(\bd)}$ which is formed from $\bx$ by deleting each $x_i$ if $d_i=1$. Now, in each channel, we apply a unique deletion vector of weight at most $t$ to the transmitted word $\bx$. Moreover, there are exactly $\binom{n}{t}$ possible deletion vectors when exactly $t$ errors occur and $V_2(n,t)$ possible deletion vectors when at most $t$ errors occur. Furthermore, in some cases, when we apply distinct deletion vectors $\bd$ and $\bd'$ to word $\bx$, we may obtain the same output word $\by$ possibly leading to a multiset of output words $Y_m$. When the set of output words can be a multiset, we will call the model \textit{multiset channel model} or \textit{multiset error pattern model}. 

In particular, the case with a unique error pattern in each channel can be considered as a \textit{generalization} of the traditional model  with unique output words. Indeed, if each output word is unique, then we have applied a unique error pattern in every channel. Moreover, if we assume that we have an output (multi)set $Y_m$ in which a different set of errors has occurred to every output word and two output words can be identical, then we could just prune this multiset $Y_m$ into a non-multiset by removing the extra copies. In other words, the error pattern model also contains the information we have in the traditional model.

Compared to the situation of the traditional model in Example~\ref{exDelTrad}, the concept of  unique deletion vectors in each channel gives new  insights and benefits  for these problems as we can see in the following example. 

\begin{example}\label{ExMulDelVec}
If we consider $\bx=111100$ and $\bx'=111010$ from Example \ref{exDelTrad} and apply different deletion vectors of weight exactly one to them, then the \textit{multiset} of output words we obtain from $\bx$ is $Y=\{11100,11100,11100,11100,11110,11110\}$ while the multiset of output words $Y'$ we obtain from $\bx'$ is $Y'=\{11010,11010,11010,11110,11100,11101\}$. Since the multisets differ, in this case we can now clearly distinguish between $\bx$ and $\bx'$ unlike in the traditional model.
Moreover, we could in fact verify that the output word multisets $Y$ and $Y'$ are unique for $\bx$ and $\bx'$, respectively (this follows from Theorem \ref{thm:delVecnAllnbound} since they have size at least five and $V_2(6,1)-V(2,1)+1=5$). 

Furthermore, the unique deletion vectors in our multiset model seem more natural when we consider this problem from a probabilistic perspective. Indeed, if we only assume that each \emph{unique} output word has equal probability to occur, then with $\bx=111100$ and $t=1$ both output words $11100$ and $11110$ have equal probability of $50\%$ to occur. However, if each deletion vector of weight one has equal probability, then $11100$ has $\approx 67\%$ probability  and $11110$ has $\approx 33\%$ probability which seems a more natural result.\end{example}

We may ask how many channels $N'$ do we probably require to obtain $N$ ($N\leq N'$) output words which have been modified by different error patterns. This probability depends on the error types occurring in the channels, on $n,N$ and $t$ for deletion errors and also on $q$ for non-deletion errors. 
Furthermore, when these parameters are suitable, it is probable that (almost) every channel has a unique error pattern and
more likely than each output being unique since there are usually more possible error patterns than output words. 
For example, when $n=N=100$ and $t=3$, we can expect roughly 100 (more precisely, $99.97...$) unique error patterns for deletion errors by Equation (\ref{Eq:ExpectedValueForUniqueEpatterns}). We will come back to exact probabilities later in Section \ref{Sec:Probabilities}. 

Besides the multiset model which requires that each error pattern occurring in a channel is unique, we also give another 
approach which does not cause challenges if some error patterns are identical as long as a predefined number of different error patterns occurs.
 We may assume that a unique set of errors occurs in each channel but instead of considering the multiset of output words we consider the set of output words. We call this model a \emph{non-multiset error pattern} model. 
 For this non-multiset model to work, we only require that in $N'$ channels we have unique error patterns, that is, we may utilize $N_s\geq N'$ channels and have identical error patterns in some of the $N_s$ channels as long as there are at least $N'$ unique error patterns. Furthermore, we \textit{do not} need to know which channels give the unique error patterns. Hence, we may use probabilistic approximation which states that if we have at least $N_s$ channels, then we are likely to have $N'$ unique error patterns. We discuss  these probabilities in greater detail in Section \ref{Sec:Probabilities} and see that our models can be confidently deployed for a wide range of parameters. 
 
To summarize the three channel models, in the \textbf{traditional channel model} each channel gives a unique output word and we have $|Y|=N$. In the \textbf{multiset error pattern channel model}, a unique set of errors occurs in each channel and we have $|Y_m|=N$ for a multiset $Y_m$ of output words. In the \textbf{non-multiset error pattern channel model}, a unique set of errors occurs in each channel but we consider a non-multiset $Y$ of output words. In particular, we may have $|Y|<N$. Notice that although in each of these cases we use the term '\textit{channel}', the meaning of the term channel is different between different models.
 \smallskip

In the following example, we compare the three channels models.

\begin{example}\label{ExComparison}
Consider a situation with $\bx=11101$ and $\bx'=11011$ when exactly $t=2$ deletion errors occur in a channel. The output words which can be obtained from both $\bx$ and $\bx'$ are $Y'=\{\by_1,\by_2,\by_3\}=\{111, 110,101\}$.  We have presented these output words together with all the deletion vectors that may lead to them in Table \ref{TabEvectEx}.

When we consider these words together with the traditional Levenshtein's model, we notice that if $\bx$ is transmitted, then we can never distinguish between $\bx$ and $\bx'$ since every output word which can be obtained from $\bx$, can also be obtained from $\bx'$. However, with the non-multiset and multiset error pattern models we can distinguish between these two words. The non-multiset model requires  (in the worst case) ten channels since we may obtain the output words in $Y'$ with nine deletion vectors from $\bx'$. Furthermore, the multiset model requires  (in the worst case) nine channels to distinguish between $\bx$ and $\bx'$ since we may obtain $\by_1$ with four deletion vectors, $\by_2$ with one deletion vector and $\by_3$ with three deletion vectors from both $\bx$ and $\bx'$, that is, in total with eight deletion vectors from both words.

Notice that in the multiset model we have more information available for us than in the non-multiset model. Hence, it is to be expected that the multiset model requires less channels than the non-multiset model.

\begin{table}[h]
\caption{The three output words $\y_i$ which we can obtain from both $\bx$ and $\bx'$ when exactly two deletion errors occur in each channel together with the deletion vectors that may lead to them.}\label{TabEvectEx}\centering
\begin{tabular}{|l||ll||ll|}
\hline
 & \centering $\bx=11101$          &  & \centering $\bx'=11011$         &  \\ \hline\hline
$\by_1=111$ & \multicolumn{1}{l|}{10010}   & 01010 & \multicolumn{1}{l|}{10100}   & 01100 \\ \hline
 & \multicolumn{1}{l|}{00110}   & 00011 & \multicolumn{1}{l|}{00110}   & 00101 \\ \hline\hline
$\by_2=110$ & \multicolumn{1}{l|}{10001} & 01001 & \multicolumn{1}{l|}{00011} &  \\ \hline
 & \multicolumn{1}{l|}{00101} &  & \multicolumn{1}{l|}{} &  \\ \hline\hline
$\by_3=101$ & \multicolumn{1}{l|}{11000} & 10100 & \multicolumn{1}{l|}{10010} & 10001 \\ \hline
 & \multicolumn{1}{l|}{01100} &  & \multicolumn{1}{l|}{01010} & 01001 \\ \hline
\end{tabular}
\end{table}
\end{example}

In Section \ref{sec:delVec}, we will consider \textit{extremal wordpairs} for deletion errors, that is, wordpairs which can be distinguished and which require the largest number of channels for distinguishing them. 
In particular, as long as $n\geq 2t+2$ (see Theorem \ref{thm:delVecnAllnbound}), there \emph{always} exists (unlike in the traditional model) a number of channels which is enough for distinguishing any wordpairs in non-multiset and multiset error pattern models. 
As we will see in Section \ref{sec:delVec} in comparison to \cite{levenshtein2001efficient},  the set of extremal wordpairs requiring the largest number of channels to distinguish them  differs in the case of deletion errors in the three models: the traditional Levenshtein's model, the multiset model and the non-multiset model. From the viewpoint of worst-case analysis, this seems interesting. Furthermore, we also show that in each of these cases the wordpair leading to the \textit{worst} case is different (see Remark \ref{RemComparison}). 
To separate between error patterns with and without multisets, we will use notations $N_m$  for the number of channels  and $\LL_m$  for the list size when we consider the multiset case. 

\medskip

Similarly to deletion errors, we could also introduce substitution vectors for \textit{substitution} errors. However, unlike in the case of deletion errors, two distinct substitution vectors would lead to distinct output words. In particular, if we know the transmitted word and the output word as well as that only substitution errors have occurred, then we can exactly deduce which substitution errors have occurred. Hence, in  the case of substitution errors it does not matter if we assume that every channel gives a unique output or if a unique set of substitution errors occur in a channel since both approaches lead to the same conclusion.

With an \textit{insertion vector} we mean an ordered set of $n+1$ ($q$-ary) words of total length at most $t$. We denote the empty word by $\varepsilon$. When an insertion error occurs, we insert, for each $1\leq i\leq n+1$, the $i$th word of the insertion vector after the $(i-1)$th symbol of the word $\bx$. Note that for $i=1$ by saying \emph{after the $0$th symbol} we mean before the first symbol.   Insertion errors are not as problematic when we assume that each channel outputs a different word. In fact, the size of an error ball of radius $t$ does not depend on the central word \cite{levenshtein2001efficient}. However, we still have some problems with the probabilities. Consider, for example, $\bx=000\in \F_2^3$ and exactly $t=1$ insertions. If we now assume (as in the traditional model) that each output word is unique, then $Y=\{0000,1000,0100,0010,0001\}$. However, if we assume that each channel has a unique error pattern, then $Y=\{0000,0000,0000,0000,1000,0100,0010,0001\}$. As we can see, the probability that we output the word $0000$ is $20\%$ in the first case and $50\%$ in the second case. Moreover, it seems natural that $0000$ is more likely than the other words to be outputted since there are four different ways to obtain it while the other words have only one. Let us denote by $B^I_t(\bx)$ the insertion ball for insertion vectors of radius $t$ centered at a word $\bx\in \F^n_q$. We have 
\begin{equation}\label{insBall}
|B^I_t(\bx)|=\sum_{i=0}^tq^i\binom{n+i}{i} .
\end{equation}
Indeed, let us consider the number of words which we can obtain from $\bx$ with exactly $j$ insertions such that $0\leq j\leq t$. Insertion vector consists of $n+1$ (possibly empty) words with total length of $j$. The answer to the question asking how many combinations there are for possible locations of inserted symbols is given by a classical combinatorial technique of \textit{stars and bars} \cite{levin2021discrete}. Indeed, this problem can be considered as having $n+1$ boxes and $j$ balls where the balls are inserted to the boxes. Thus, the technique of stars and bars tells that there are $\binom{n+j}{j}$ ways to insert balls into these boxes. Moreover, there are $q^j$ ways to choose the inserted symbols once their locations are known. Notice that the cardinality in \eqref{insBall} differs from the cardinality of an insertion ball considered in \cite{levenshtein2001efficient}.

\subsection{Probabilities}\label{Sec:Probabilities}
In this section, we briefly consider the probabilities on how many unique error patterns we might obtain when each error pattern is equally probable.
Consider a setup in which we have a collection of $m$ distinct coupons and each time we draw one coupon, it is replaced with a new one.
In a well-known \textit{Coupon Collector Problem} \cite{flajolet1992birthday}, we are asked how many coupons we need to buy randomly to collect at least one copy of every coupon. Furthermore, in the \textit{Partial Coupon Collector Problem} PCCP$(j,m)$, we are asked how many coupons we need to buy randomly to obtain $j$ different coupons from $m$ total coupon types. This setup corresponds to our problem with distinct error patterns in channels assuming that all error patterns are equally likely. Asking: ``Through how many channels do we need to transmit word $\bx$ to obtain $j$ distinct error patterns" is the same as PCCP$(j,m)$. In \cite{flajolet1992birthday}, the expected value for PCCP$(j,m)$ has been presented as:
 $$E[PCCP(j,m)]=m(H_m-H_{m-j})\approx m\ln\frac{m}{m-j},$$
where $H_m$ is the $m$th harmonic number $\sum_{i=1}^m\frac{1}{i}$. The approximation follows from  $\gamma + 1/(2m+2) + \ln m < H_m <\gamma + 1/(2m) + \ln m$, where $\gamma$ is the Euler-Mascheroni constant, see \cite{young199175}. When we consider deletion vectors and at most $t_d$ deletions occur in any channel, the value $m$ is $V_2(n,t_d)$. For insertion errors with at most $t_i$ insertions in a channel, the value $m$ is $|B_{t_i}^I(\bx)|=\sum_{j=0}^{t_i}q^j\binom{n+j}{j}$ by Equation (\ref{insBall}).

 Another way to consider this problem is: If we transmit the word $\bx\in \F^n_q$ through $N$ channels and one of $m$ equally likely error patterns may occur in any of them, what is the expected number of unique error patterns occurring in these channels? Observe that the likelihood of any single error pattern not occurring in any of the $N$ channels is $(\frac{m-1}{m})^N$. 
 Thus, a particular error pattern occurs in at least one of the channels with probability $1-(\frac{m-1}{m})^N$ and hence, the expected number of unique error patterns is  \begin{equation}\label{Eq:ExpectedValueForUniqueEpatterns}
 m\left(1-\left(\frac{m-1}{m}\right)^N\right).
 \end{equation}

\subsection{Deletion vectors}\label{sec:delVec}
In this subsection, we consider how many channels we may require to ensure that we can uniquely determine the transmitted word, that is $\LL=1$, when we have deletion vectors of weight at most (or exactly) $t$. 
We give two types of results. First we consider the non-multiset error pattern model and give the exact minimum number of different error patterns, that is, the minimum value for $N$ which guarantees different \textit{sets} of output words from two different transmitted words, that is, cases where $\LL=1$. 
Then, we consider the same problem for the \emph{multiset} model. 
In the following theorem, we provide a number of channels, which guarantees that $\LL=1$ in the non-multiset error pattern case. However, observe that Theorem \ref{thm:delVecnAllnbound} also holds for the multiset model and gives $\LL_m=1$, since if the sets of output words are different, then also multisets of the output words are different. In this section, we restrict our considerations to $C=\F^n_q$. In the following theorem we concentrate on the case with $q=2$. However, the same result holds also for larger $q$ (see the discussion after Theorem \ref{thm:delVecnAllnbound}).

\begin{theorem}\label{thm:delVecnAllnbound}
Let $t$, $n$ and $q$ be integers such that $t \geq 1$, $n \geq 2t+2 \geq 4$ and $q=2$.

\begin{enumerate}
\item[(i)] If at most $t$ deletion errors occur in a channel and the number of channels $N\geq V_2(n,t)-V_2(\lceil n/2\rceil-1,t)+1$, then the output word (multi)set $Y$ is unique for any transmitted word and $\LL=1$.
\item[(ii)] If exactly $t$ deletion errors occur in a channel and the number of channels $N\geq \binom{n}{t}-\binom{\lceil n/2\rceil-1}{t}+1$, then the output word (multi)set $Y$ is unique for any transmitted word and $\LL=1$.
\end{enumerate}
\end{theorem}

\begin{proof}
In this proof, we only consider Case (i). However, the proof for (ii) follows by replacing  in the following proof each $V_2(a,b)$ by $\binom{a}{b}$  and by changing every weight constraint $w(\bd)\le t$ to $w(\bd)=t$ for a deletion vector $\bd$.

Let $N\geq V_2(n,t)-V_2(\lceil n/2\rceil-1,t)+1$. Note that we may apply $V_2(n,t)-V_2(\lceil n/2\rceil-1,t)+1$ unique deletion vectors on any word of length $n$ as the number of all such vectors is equal to $V_2(n,t)$. 
Suppose on the contrary that there exists a set of output words $Y$ which can be obtained by applying a set $D$ of $N$ deletion vectors to $\bx$ and also by applying  a set $D'$ of $N$ deletion vectors to $\bx'$. We first show that both $\bx$ and $\bx'$ have the same weight.

\noindent \textbf{Claim 1:} We have $w(\bx)=w(\bx')$. 

\noindent\textit{Proof of Claim 1.}
Let us suppose on the contrary, without loss of generality, that $w(\bx)>w(\bx')$. Moreover, let us first assume that $w(\bx)> n/2$. 
We denote $m=w(\bx)-w(\bx')$.

Let us denote by $D''$ a set of deletion vectors $\bd''$ with $w(\bd'')\leq t$, where we can have $d_i''=1$ if $x'_i=1$ or for at most $m-1$ indices $i$ for which $x'_i=0$ (in other words, deletion vectors delete some $1$'s and at most $m-1$ symbols $0$ from $\bx'$).  Then, we have $|D''|\geq V_2(m-1+w(\bx'),t)=V_2(w(\bx)-1,t)\geq V_2(\lfloor n/2\rfloor,t)$. Furthermore, we can observe that if we obtain $\by'$ from $\bx'$ with a deletion vector in $D''$, then $\by'$ has at least $n-w(\bx')-(m-1)=n-w(\bx)+1$ zeroes and we cannot obtain $\by'$ from $\bx$ with any deletion vector since $\bx$ has $n-w(\bx)$ zeroes. Thus, \begin{equation}\label{eq:difWeight}
|D'|\leq V_2(n,t)-|D''|\leq V_2(n,t)-V_2(\lfloor n/2\rfloor,t)<V_2(n,t)-V_2(\lceil n/2\rceil-1,t)+1\leq N,\end{equation} a contradiction. 

Moreover, the case $w(\bx)\leq n/2$ is similar. Indeed, in this case we may swap 1's and 0's in $\bx$ and $\bx'$. Now $w(\bx')> w(\bx)$ and $w(\bx')>n/2$. After this, we could apply above proof by swapping $\bx$ and $\bx'$, by switching $D'$ to $D$ and $m$ to $m'=w(\bx')-w(\bx)$.
Thus, Claim $1$ follows.

\medskip

Let us now assume that $h$ is the smallest index for which $x_h\neq x'_h$. Furthermore, without loss of generality, assume that $x_h=0$ and $x'_h=1$. Moreover, let us notate $w_1(\bw)=|\{i\mid i\leq h,i\in \supp(\bw)\}|$ and $w_2(\bw)=|\{i\mid i> h,i\in \supp(\bw)\}|$. We have $w_1(\bx)=w_1(\bx')-1$ and $w_2(\bx)=w_2(\bx')+1$, since $w(\bx)=w(\bx')$. 
There are $w_1(\bx)$ symbols $1$ and $h-w_1(\bx)$ symbols $0$ in $\bx$ before the $(h+1)$th coordinate. Moreover, in $\bx'$, before the $(h+1)$th coordinate, there are $w_1(\bx)+1$ symbols $1$ and $h-w_1(\bx)-1$ symbols $0$.

Let us consider the following deletion vector sets $D_1\subseteq \F_2^n$ and $D_2\subseteq \F_2^n$:
\begin{align*}
D_1=&\{\bd\in \F^n_2\mid w(\bd)\leq t, \supp(\bd)\cap[1,h]\subseteq \supp(\bx)\\
 &\text{and } (\supp(\bd)\cap[h+1,n])\cap \supp(\bx)=\emptyset\} \text{ and}\\
%
%
D_2=&\{\bd\in \F^n_2\mid w(\bd)\leq t, \supp(\bd)\cap[h+1,n]\subseteq \supp(\bx')\\
 & \text{and } (\supp(\bd)\cap[1,h])\cap \supp(\bx')=\emptyset\}.\end{align*}
In other words, we have $\bd\in D_1$ if and only if $w(\bd)\leq t$ and its support is such that 
for $d_i = 1$ we require that $x_i = 1$ and $i \leq h$, or $x_i = 0$ and $i > h$.
Similarly we have $\bd'\in D_2$ if and only if $w(\bd')\leq t$ and its support is such that for $d_i' = 1$ we require that $x_i' = 0$ and $i \leq h$, or $x_i' = 1$ and $i > h$.

\medskip

\noindent \textbf{Claim 2:} If we obtain $\by$ from $\bx$ with a deletion vector $\bd\in D_1$, then $\by$ cannot be obtained with any deletion vector from $\bx'$.

\noindent \textit{Proof  of Claim 2.} Let us assume that $\by$ is obtained from $\bx$ with a deletion vector $\bd\in D_1$. Suppose on the contrary that we can obtain $\by$ from $\bx'$ with some deletion vector $\bd'$ of weight at most $t$. We have $w(\bd)=w(\bd')$ as the original words $\bx$ and $\bx'$ are of equal length. 

Let us first consider how $\by$ can be obtained from $\bx$ with $\bd$. Consider the symbol $0$ in the $h$th coordinate (note that it is not deleted by $\bd$). In $\bx$, there are $h-w_1(\bx)-1$ symbols $0$ before it and $w_1(\bx)$ symbols $1$. Notice that $\bd$ deletes symbols $0$ from $\bx$ only after the symbol $0$ in the coordinate $h$. Hence, when we consider the symbol $0$ in $\by$ which has $h-w_1(\bx)-1$ symbols $0$ before it, we notice that it has exactly $w_1(\bx)-w_1(\bd)$ symbols $1$ before it.

Let us then consider how $\by$ can be obtained from $\bx'$ with $\bd'\in \F^n_2$. Consider the first symbol $0$ after the $h$th coordinate in $\bx'$. This symbol exists because $w(\bx)=w(\bx')$. In $\bx'$, there are $h-w_1(\bx)-1$ symbols $0$ before it and at least $w_1(\bx)+1$ symbols $1$. Since $w(\bx)=w(\bx')$ and $\bd$ deletes exactly $w_1(\bd)$ symbols $1$, we can delete at most $w_1(\bd)$ symbols $1$ from $\bx'$ with $\bd'$. Thus, there are at least $w_1(\bx)+1-w_1(\bd)$ symbols $1$ before the symbol $0$  which has $h-w_1(\bx)-1$ symbols $0$ before it in $\by$ obtained from $\bx'$. Notice that this kind of symbol $0$ must exist also in $\by$ when we obtain it from $\bx'$ since $\by$ has at least $h-w_1(\bx)$ symbols $0$ when we obtain it from $\bx$. Thus, word $\by$, which we obtained from $\bx$, is not identical with word $\by$ which we obtained from $\bx'$, a contradiction which proves Claim 2.
\medskip


Since $D_1$ and $D_2$ are constructed in symmetrical ways, Claim 2 also holds for $D_2$, that is, if we obtain $\by$ from $\bx'$ with $\bd'\in D_2$, then $\by$ cannot be obtained with any deletion vector from $\bx$. Indeed, to prove this, we consider the $w_1(\bx')-1$ symbols $1$ in $\bx'$ before the $h$th coordinate and then we  construct $\by$ which has $h-w_1(\bx')-w_1(\bd')$ symbols $0$ before the $w_1(\bx')$th symbol $1$. When we try to obtain this $\by$ from $\bx$, we notice that there are always at least $h-w_1(\bx')-w_1(\bd')+1$ zeroes before the $w_1(\bx')$th symbol $1$ (if it exists) since $\bd$ can delete at most $w_1(\bd')$ symbols $0$  from $\bx$ (recall that $w(\bx)=w(\bx')$) and there are, in $\bx$, $h-w_1(\bx')+1$ symbols $0$ before the $w_1(\bx')$th symbol $1$.

\noindent \textbf{Claim 3:} We have $|D_i|\geq V_2(\lceil n/2\rceil-1,t)$ for $i=1$ or $i=2$.

\noindent \textit{Proof  of Claim 3.}
There are $w_1(\bx)$ symbols $1$ and $n-h-w_2(\bx)$ symbols $0$ which we can remove with a deletion vector $\bd\in D_1$. Thus, $|D_1|= V_2(n-h-w_2(\bx)+w_1(\bx),t)$. Similarly, there are $h-w_1(\bx')=h-w_1(\bx)-1$ symbols $0$ and $w_2(\bx')=w_2(\bx)-1$ symbols $1$ which we can remove with deletion vector $\bd'\in D_2$. Thus, $|D_2|= V_2(h-w_1(\bx)-1+w_2(\bx)-1,t)$.

We split the proof between Cases A) $|D_1|\geq|D_2|$ and B) $|D_2|\geq|D_1|$.
Consider first Case A). 
Hence, $n-h-w_2(\bx)+w_1(\bx) \geq h-w_1(\bx)-1+w_2(\bx)-1$. Notice that $n-h-w_2(\bx)+w_1(\bx)+ (h-w_1(\bx)-1+w_2(\bx)-1)=n-2$. Since $|D_1|\geq|D_2|$, we have $n-h-w_2(\bx)+w_1(\bx)\geq \lceil\frac{n-2}{2}\rceil$. Thus, $$|D_1|=V_2(n-h-w_2(\bx)+w_1(\bx),t)\geq V_2(\lceil n/2\rceil-1,t)$$ as claimed. 

Case B) is similar. We have $n-h-w_2(\bx)+w_1(\bx) \leq h-w_1(\bx)-1+w_2(\bx)-1$ and $n-h-w_2(\bx)+w_1(\bx)+ (h-w_1(\bx)-1+w_2(\bx)-1)=n-2$. Since $|D_2|\geq|D_1|$, we have $h-w_1(\bx)-1+w_2(\bx)-1\geq \lceil\frac{n-2}{2}\rceil$. Thus, $$|D_2|=V_2(h-w_1(\bx)-1+w_2(\bx)-1,t)\geq V_2(\lceil n/2\rceil-1,t)$$ as claimed. 
Now, Claim $3$ follows.
\medskip

Let $\{v,w\}=\{1,2\}$ and $|D_v|\geq |D_w|$. If $v=1$, then $\bx_v=\bx$ and $\bx_w=\bx'$. If $v=2$, then $\bx_w=\bx$ and $\bx_v=\bx'$.
We have $|D_v|\geq V_2(\lceil n/2\rceil-1,t)$ by Claim $3$ and by Claim $2$ we cannot obtain the output words $\by$, which are obtained  from $\bx_v$ with $\bd\in D_v$,   with any deletion vector from $\bx_w$. Thus, we have $N=|D|\leq V_2(n,t)-|D_v|\leq V_2(n,t)-V_2(\lceil n/2\rceil-1,t)$. Therefore, $\LL= 1$ when $N\geq V_2(n,t)-V_2(\lceil n/2\rceil-1,t)+1$. 
\end{proof}

Observe that although we consider the binary case in the previous theorem, it is easy to see that the result applies also for larger alphabets. Indeed, consider distinct words $\bx,\bx'\in \F_q^n$ for $q>2$, and let $i\in[1,n]$ be such that $x_i\neq x_i'$. We partition $\F_q$ into non-empty sets $A$ and $B$ such that $x_i\in A$ and $x_i'\in B$, and let $f:\F_q^n\to\F^n_2$ transform any $q$-ary word to binary word by changing all symbols from $A$ to zeroes and symbols from $B$ to ones.  For example, we may choose $A = \{x_i\}$ and $B = [0,q-1] \setminus \{x_i\}$. Note that $f(\bx)\neq f(\bx')$. Now, we may observe that if a deletion vector $\bd$ turns $\bx$ and $\bx'$ into the same word $\by$, then $\bd$ transforms both $f(\bx)$ and $f(\bx')$ into $f(\by)$. 
Hence, the lower bound of Theorem \ref{thm:delVecnAllnbound} applies also for larger alphabets.

In the subsequent theorem, we see that the lower bounds of Theorem~\ref{thm:delVecnAllnbound}  are tight (also for larger alphabets) in the non-multiset error pattern model. This is done by showing that a suitably chosen pair $\bx,\bx'$ form an extremal wordpair for the non-multiset model.

\begin{theorem}\label{thm:delVecnAllNormalSet}
Let $n\geq2t+2\geq4$ and $C=\F^n_q$. Consider words $\bx=0^{\lceil n/2\rceil-1}10^{\lfloor n/2\rfloor}$ 
and $\bx'=0^{\lceil n/2\rceil}10^{\lfloor n/2\rfloor-1}$.
\begin{enumerate}
\item[(i)] If at most $t$ deletion errors occur in a channel and the number of channels $N\leq V_2(n,t)-V_2(\lceil n/2\rceil-1,t)$, then we can obtain the same output word set $Y$  from $N$ channels with input words $\bx$ and $\bx'$.
\item[(ii)] If exactly $t$ deletion errors occur in a channel and the number of channels $N\leq \binom{n}{t}-\binom{\lceil n/2\rceil-1}{t}$, then we can obtain the same output word set $Y$  from $N$ channels with input words $\bx$ and $\bx'$.
\end{enumerate}
%
\end{theorem}
\begin{proof}
Again, in this proof, we only consider Case (i) and Case~(ii) can be shown by replacing  in the following proof each $V_2(a,b)$ by $\binom{a}{b}$  and by changing every weight constraint $w(\bd)\le t$ to $w(\bd)=t$ for a deletion vector $\bd$.

Let us transmit a word $\bx\in\F^n_q$ with $\supp(\bx)=\{\lceil n/2\rceil\}$ and $x_{\lceil n/2\rceil}=1$ through $N$ channels. It is enough to consider $N=V_2(n,t)-V_2(\lceil n/2\rceil-1,t)$ channels. Let $\bx'\in \F^n_q$ be a word such that $\supp(\bx')=\{\lceil n/2\rceil+1\}$ and $x_{\lceil n/2\rceil+1}'=1$. Let us consider following sets of deletion vectors: 
$$D=\{\bd\mid w(\bd)\leq t, \supp(\bd)\cap[\lceil n/2\rceil,n]\neq \emptyset\}$$ 
and  (notice that $\lceil n/2 \rceil - \lfloor n/2 \rfloor$ is 0 or 1 depending on the parity of $n$)
$$D'=\{\bd\mid w(\bd)\leq t, \supp(\bd)\cap[1+(\lceil n/2\rceil-\lfloor n/2\rfloor),\lceil n/2\rceil+1]\neq \emptyset\}.$$

Observe that $|D|=|D'|=V_2(n,t)-V_2(\lceil n/2\rceil-1,t)$. Indeed, let us consider first the set $D$. There are $V_2(n,t)$ different vectors of weight at most $t$. Moreover, those vectors belong to $D$ unless their support is within the set $[1,\lceil n/2\rceil-1]$ and there are $V_2(\lceil n/2\rceil-1,t)$ such vectors. 
Similarly, it can be shown that $|D'| = V_2(n,t) - V_2(\lceil n/2\rceil-1,t)$.

Let us first consider the set of output words $Y$ which we obtain from $\bx$ with $D$. We have 
\begin{align*}
Y=&\{\by\in\{0,1\}^h\mid n-t\leq h\leq n-1, w(\by)\leq1\\
&\text{and } \supp(\by)\subseteq [h+1-\lfloor n/2\rfloor, \lceil n/2\rceil]\}.
\end{align*}
Indeed, we delete at least one and at most $t$ symbols and hence, $\by\in \{0,1\}^h$ for  $n-t\leq h\leq n-1$. Moreover, $w(\by)\leq 1$, since it is possible that we delete the only symbol $1$ in $\bx$. Finally, if $w(\by)=1$, then we deleted at least one symbol $0$ after the symbol $1$ in $\bx$ and hence, $\supp(\by)\subseteq [h-(\lfloor n/2 \rfloor - 1), \lceil n/2\rceil]$. In other words, there are at most $\lfloor n/2\rfloor-1$ symbols $0$ after the symbol $1$.

Let us then consider the set of output words $Y'$ which we obtain from $\bx$ with $D'$. Similarly to previous case, it is easy to check that $Y'=Y$.
%
Thus, the claim follows.\end{proof}

Let us next consider deletion vectors in the \emph{multisets} model. Now we use multisets of output words to distinguish between two possible transmitted words.
Recall that the lower bounds of Theorem~\ref{thm:delVecnAllnbound} hold also for the multiset model. 
Theorem \ref{thm:delVecnAllnbound}(i) is tight in the case of multisets when $t=1$. 

\begin{proposition}\label{DelVecMult1}
Let $t=1$ and $q\geq2$. If $N_m\leq \lceil n/2\rceil+1=V_2(n,1)-V_2(\lceil n/2\rceil-1,1)$, then $\LL_m\geq2$.
\end{proposition}
\begin{proof}
Let us transmit the words $\bx,\bx'\in\F_q^n$ with $\supp(\bx)=\{\lceil n/2\rceil\}$, $x_{\lceil n/2\rceil}=1$ and $\supp(\bx')=\{\lceil n/2\rceil+1\}$, $x_{\lceil n/2\rceil+1}'=1$ through $N_m=V_2(n,t)-V_2(\lceil n/2\rceil-1,t)=\lfloor n/2\rfloor+1$ channels. The claim follows by applying the sets of deletion vectors $D=\{\bd_i\mid  \supp(\bd_i)=\{n+1-i\}, i\in[1,\lfloor n/2\rfloor+1]\}$ and $D'=\{\bd'_i\mid \supp(\bd'_i)=\{i+\lceil n/2\rceil-\lfloor n/2\rfloor, i\in[1,\lfloor n/2\rfloor+1]\}\}$ to $\bx$ and $\bx'$, respectively. Indeed, we get the same multiset $Y_m$ (consisting of   $\lfloor n/2 \rfloor$ times the word $0^{\lceil n/2 \rceil -1}10^{\lfloor n/2 \rfloor-1}$ and once the word $0^{n-1}$) in both cases.
%
\end{proof}

Next we provide further results for the multiset model and compare it with the non-multiset one. We have decided to focus on even $n$ in the upcoming considerations to avoid extra complications, as the behaviour of the odd $n$ seems to be somewhat different.  Note that, in the following, we sometimes  concentrate on the case in which exactly $t$ errors occur in channels instead of a case in which at most $t$ errors occur.

As we can see in the following proposition, corollaries and especially in Remark \ref{Remt=2}, the construction which gives a tight bound for the case with $t=1$ \textit{does not} give a tight bound for $t=2$ for the multiset model.

\begin{proposition}\label{DelVecMultw1}
Let $t\geq1$, $q\geq2$ and $n\geq t+1$. Let us assume that exactly $t$ deletion errors occur in a channel in the multiset model. We can distinguish between
$\bx=0^{a-1}10^{n-a}$ and $\bx'=0^{a}10^{n-a-1}$
%
with $N=\binom{n}{t}-\binom{a-1}{\lfloor\frac{at+a}{n}\rfloor}\binom{n-a-1}{t-\lfloor\frac{at+a}{n}\rfloor}+1$ channels while $N-1$ is not enough.\end{proposition}
\begin{proof}
Let us transmit the words $\bx,\bx'\in\F_q^n$, for which $\supp(\bx) = \{a\}$, $\supp(\bx') = \{a+1\}$ and $x_a = x'_{a+1} = 1$, 
 through $N_m=\binom{n}{t}-\binom{a-1}{\lfloor\frac{at+a}{n}\rfloor}\binom{n-a-1}{t-\lfloor\frac{at+a}{n}\rfloor}$ channels. 

First of all, we notice that we can obtain the word $\nolla\in \F_q^{n-t}$ from both $\bx$ and $\bx'$ with $\binom{n-1}{t-1}$ deletion vectors each deleting the single $1$. Furthermore, we can obtain a word $\by_1$ with $\supp(\by_1)=\{a-i\}$ from word $\bx$ with $\binom{a-1}{i}\binom{n-a}{t-i}$ deletion vectors for $i\in[0,t-1]$ and from $\bx'$ with $\binom{a}{i+1}\binom{n-a}{t-i-1}$ deletion vectors for $i\in[0,t-1]$. Note that the upper bound $t-1$ for $i$ is due to the fact that we cannot obtain the word $\by_1$ with $\supp\{a-t\}$ from $\bx'$. Together, these mean that we may obtain the same multiset of output words from $$\binom{n-1}{t-1}+\sum_{i=0}^{t-1}\min\left\{\binom{a}{i+1}\binom{n-a-1}{t-i-1},\binom{a-1}{i}\binom{n-a}{t-i}\right\}$$ channels.

Let us consider $\binom{a}{i+1}\binom{n-a-1}{t-i-1}/\left(\binom{a-1}{i}\binom{n-a}{t-i}\right)$ for $i\in[0,t-1]$ to determine when one of these is smaller. We have
$$
	\frac{\binom{a}{i+1}\binom{n-a-1}{t-i-1}}{\binom{a-1}{i}\binom{n-a}{t-i}} \geq 1  \hspace{0.5cm}
	\Leftrightarrow \hspace{0.5cm}\frac{a(t-i)}{(i+1)(n-a)} \geq 1  \hspace{0.5cm}
	\Leftrightarrow \hspace{0.5cm}\frac{at+a}{n}-1 \geq i.
$$

Let us denote $A=\lfloor\frac{at+a}{n}\rfloor-1$. In the following, we use the well-known binomial identities, of which the first is Vandermonde's identity, $\sum_{i=0}^k\binom{r}{i}\binom{p}{k-i}=\binom{p+r}{k}$ and  $\binom{n}{k}=\binom{n-1}{k}+\binom{n-1}{k-1}$. Now we have \begin{align*}
&\binom{n-1}{t-1}+\sum_{i=0}^{t-1}\min\left\{\binom{a}{i+1}\binom{n-a-1}{t-i-1},\binom{a-1}{i}\binom{n-a}{t-i}\right\}\\
=&\binom{n-1}{t-1}+\sum_{i=0}^{A}\binom{a-1}{i}\binom{n-a}{t-i}+\sum_{i=A+1}^{t-1}\binom{a}{i+1}\binom{n-a-1}{t-i-1}\\
=&\binom{n-1}{t-1}+\sum_{i=0}^{A}\binom{a-1}{i}\binom{n-a-1}{t-i}+\sum_{i=0}^{A}\binom{a-1}{i}\binom{n-a-1}{t-i-1}\\
+&\sum_{i=A+1}^{t-1}\binom{a-1}{i+1}\binom{n-a-1}{t-i-1}+\sum_{i=A+1}^{t-1}\binom{a-1}{i}\binom{n-a-1}{t-i-1}\\
=&\binom{n-1}{t-1}+\sum_{i=0}^{t-1}\binom{a-1}{i}\binom{n-a-1}{t-i-1}+\sum_{i=0}^{A}\binom{a-1}{i}\binom{n-a-1}{t-i}\\
+&\sum_{i=A+2}^{t}\binom{a-1}{i}\binom{n-a-1}{t-i}\\
=&\binom{n-1}{t-1}+\binom{n-2}{t-1}+\sum_{i=0}^{A}\binom{a-1}{i}\binom{n-a-1}{t-i}\\
+&\sum_{i=A+1}^{t}\binom{a-1}{i}\binom{n-a-1}{t-i}-\binom{a-1}{A+1}\binom{n-a-1}{t-A-1}\\
=&\binom{n-1}{t-1}+\binom{n-2}{t-1}+\sum_{i=0}^{t}\binom{a-1}{i}\binom{n-a-1}{t-i}-\binom{a-1}{A+1}\binom{n-a-1}{t-A-1}\\
=&\binom{n-1}{t-1}+\binom{n-2}{t-1}+\binom{n-2}{t}-\binom{a-1}{A+1}\binom{n-a-1}{t-A-1}\\
=&\binom{n}{t}-\binom{a-1}{A+1}\binom{n-a-1}{t-A-1}.
\end{align*}

Hence, we require exactly $\binom{n}{t}-\binom{a-1}{A+1}\binom{n-a-1}{t-A-1}+1$ channels to distinguish between $\bx$ and $\bx'$ in the multiset model.
\end{proof}

Corollaries \ref{CorDelVecMulExBad}, \ref{CorDelVecMulExHalf} and \ref{CorDelVecMulCumHalf} follow from Proposition \ref{DelVecMultw1}. In these corollaries we establish 
how many channels are exactly required for
distinguishing between some interesting wordpairs $\bx_1,\bx_2$. These wordpairs are interesting since at least for some values of $n$ and $t$, they are extremal (see the discussion in Remark \ref{RemComparison}). However, it is possible that for some values of $n$ and $t$, these are not extremal wordpairs.

\begin{corollary}\label{CorDelVecMulExBad}
Let $n= h(2t+2)$, positive integers $h\geq1$ and $q,t\geq2$. If exactly $t$ deletions occur in each channel, then in the multiset model exactly  $$\binom{n}{t}-\binom{ht-1}{\lfloor t/2\rfloor}\binom{h(t+2)-1}{\lceil t/2\rceil}+1$$ channels is enough for distinguishing between $\bx_1=0^{ht-1}10^{h(t+2)}$ and $\bx_2=0^{ht}10^{h(t+2)-1}$.
%
%
\end{corollary}

\begin{corollary}\label{CorDelVecMulExHalf}
Let $n\geq 2t+2$ be even, $t\geq1$, $q\geq2$. If exactly $t$ deletions occur in each channel, then in the multiset model exactly  $$\binom{n}{t}-\binom{n/2-1}{\lfloor t/2\rfloor}\binom{n/2-1}{\lceil t/2\rceil}+1$$ channels is enough for distinguishing between $\bx_1=0^{n/2-1}10^{n/2}$ and $\bx_2=0^{n/2}10^{n/2-1}$.
%
\end{corollary}

Next we consider the case of \emph{at most} $t$ deletion errors in a channel. 

\begin{corollary}\label{CorDelVecMulCumHalf}
Let $n\geq 2t+2$ be even, $q,t\geq2$. If \emph{at most} $t$ deletions occur in each channel, then  in the multiset model for $a=n/2$ exactly  $$V(n,t)-\sum_{i=0}^t\binom{n/2-1}{\lfloor\frac{i}{2}\rfloor}\binom{n/2-1}{\lceil\frac{i}{2}\rceil}+1$$ channels is enough for distinguishing between $\bx_1=0^{n/2-1}10^{n/2}$ and $\bx_2=0^{n/2}10^{n/2-1}$.
%
\end{corollary}\begin{proof}
Let us denote by $N_i$, for $0\leq i\leq t$, the number of channels we require to distinguish between $\bx_1$ and $\bx_2$ when exactly $i$ errors occur. Then, we require $\sum_{i=0}^t (N_i-1)+1$ channels to distinguish between $\bx_1$ and $\bx_2$ when at most $t$ errors occur in a channel since obtaining an output word $\by\in \F^{n-i}_q$ from both $\bx_1$ and $\bx_2$ requires that exactly $i$ deletions occur to both $\bx_1$ and $\bx_2$. Notice that $N_0=1=\binom{n}{0}-\binom{n/2-1}{\lfloor 0/2\rfloor}\binom{n/2-1}{\lceil 0/2\rceil}+1$ and we obtain the other values for $N_i$ from Corollary \ref{CorDelVecMulExHalf}. Hence, we have $$\sum_{i=0}^t (N_i-1)+1=\sum_{i=0}^t \left(\binom{n}{i}-\binom{n/2-1}{\lfloor\frac{i}{2}\rfloor}\binom{n/2-1}{\lceil\frac{i}{2}\rceil}\right)+1=V(n,t)-\sum_{i=0}^t\binom{n/2-1}{\lfloor\frac{i}{2}\rfloor}\binom{n/2-1}{\lceil\frac{i}{2}\rceil}+1.$$
\end{proof}

In the following remark we observe some differences in the behaviour of extremal wordpairs between the multiset and non-multiset models.
\begin{remark}\label{Remt=2}
Let $t=2$ and $n=6h$ for an integer $h\geq2$. Let at most $t$ deletions occur in any channel. Then, by Corollary \ref{CorDelVecMulCumHalf}, the exact minimum number of channels required in the multiset model for distinguishing between $\bx_1=0^{n/2-1}10^{n/2}$ and $\bx_2=0^{n/2}10^{n/2-1}$ is $N= V(n,2)-\sum_{i=0}^2\binom{n/2-1}{\lfloor\frac{i}{2}\rfloor}\binom{n/2-1}{\lceil\frac{i}{2}\rceil}+1=\binom{n}{2}+n+1-(n/2-1)-(n/2-1)^2=n^2/4+n+1$. On the other hand, by considering Corollary \ref{CorDelVecMulExBad} with $t=2$ and Proposition \ref{DelVecMultw1} with $t=1$, the exact minimum number of channels required in the multiset model for distinguishing between $\bx=0^{2h-1}10^{4h}$ and $\bx'=0^{2h}10^{4h-1}$ is \begin{align*}
&\binom{n}{2}-(2h-1)(h(2+2)-1)+(n-\binom{2h-1}{\lfloor\frac{4h}{n}\rfloor}\binom{n-2h-1}{1-\lfloor\frac{4h}{n}\rfloor})+1\\
=&\binom{n}{2}-(8h^2-6h+1)+(n-(n-2h-1))+1\\
=&\binom{n}{2}-2n^2/9+4n/3+1=5n^2/18+5n/6+1.
\end{align*}
Hence, we require $n^2/36-n/6$ more channels to distinguish between $\bx$ and $\bx'$ than we need for distinguishing between $\bx_1$ and $\bx_2$. Recall that by Theorems \ref{thm:delVecnAllnbound} and \ref{thm:delVecnAllNormalSet}, the wordpair $\bx_1,\bx_2$ requires the most channels for the non-multiset model.
Thus, the set of extremal wordpairs differs between the multiset and non-multiset models for deletion errors for some parameters of $n$ and $t$.
\end{remark}

In the following lemma, we consider how many channels we may require for distinguishing two words whose weights differ by $b$. Furthermore, we present a pair of words attaining the presented bound.



\begin{lemma}\label{DelVecMulDifwWords}
Let exactly $t$ deletions occur in any channel in the multiset model. If $w(\bx_1)=w(\bx_2)+b$ for two binary words on symbols 0 and 1 and $b\geq1$, then the number of channels $N$ for which we may obtain the same output word set from both $\bx_1$ and $\bx_2$ is at most
$$N=\sum_{i=b}^{t}\min\left\{\binom{w(\bx_1)}{i}\binom{n-w(\bx_1)}{t-i},\binom{w(\bx_1)-b}{i-b}\binom{n-w(\bx_1)+b}{t+b-i}\right\}$$
and the value is tight for words
$\bx_1=1^{w(\bx_1)}0^{n-w(\bx_1)}$ and $\bx_2=1^{w(\bx_1)-b}0^{n+b-w(\bx_1)}$. 
%
%
\end{lemma}
\begin{proof}
Let $w(\bx_1)=w(\bx_2)+b$ for some $b\geq1$ for two binary words on symbols 0 and 1. Let us consider the multiset of output words we can obtain from both of these words. We can observe that we need to delete at least $b$ symbols $1$ from $\bx_1$ and at least $b$ symbols $0$ from $\bx_2$. Clearly, $b\le t$. Moreover, if we delete exactly $i$ symbols $1$ from $\bx_1$ and exactly $t-i$ symbols $0$ from $\bx_1$ to obtain some output word, then to obtain the resulting output word we need to delete exactly $i-b$ symbols $1$ and exactly $t+b-i$ symbols $0$ from $\bx_2$. Hence, we may obtain these output words from at most $\binom{w(\bx_1)}{i}\binom{n-w(\bx_1)}{t-i}$ channels from $\bx_1$ and from at most $\binom{w(\bx_1)-b}{i-b}\binom{n-w(\bx_1)+b}{t+b-i}$ channels from $\bx_2$. In other words, we can obtain them from at most $$N=\sum_{i=b}^{t}\min\left\{\binom{w(\bx_1)}{i}\binom{n-w(\bx_1)}{t-i},\binom{w(\bx_1)-b}{i-b}\binom{n-w(\bx_1)+b}{t+b-i}\right\}$$ channels, as claimed. Finally, we may observe that if $\bx_1=1^{w(\bx_1)}0^{n-w(\bx_1)}$ and $\bx_2=1^{w(\bx_1)-b}0^{n+b-w(\bx_1)}$, then we  have the same output word multiset for $N$ channels so the upper bound is tight.
\end{proof}

In the following proposition, we examine more closely the case from the previous lemma with $b=1$.

\begin{proposition}\label{DelVecMulDifw2}
Let $n$ be even  and let exactly $t$ deletions occur in any channel in the multiset model. If $w(\bx_1)=w(\bx_2)+1$ for two binary words on symbols 0 and 1, then there exists a pair $\bx$ and $\bx'$ of binary words on symbols 0 and 1 with $w(\bx)=w(\bx')=1$ such that we require at least as many channels for distinguishing between $\bx$ and $\bx'$ as we need for distinguishing between $\bx_1$ and $\bx_2$.
%
\end{proposition}
\begin{proof}
Let $w(\bx_1)=w(\bx_2)+1$. By Lemma \ref{DelVecMulDifwWords} with $b=1$, we may obtain the same output word multiset from both $\bx_1$ and $\bx_2$ when 
 \begin{equation}\label{eqw1eroMultDel}
N=\sum_{i=1}^{t}\min\left\{\binom{w(\bx_1)}{i}\binom{n-w(\bx_1)}{t-i},\binom{w(\bx_1)-1}{i-1}\binom{n-w(\bx_1)+1}{t+1-i}\right\}.
\end{equation}
Moreover, this is attained by the pair   $\bx_1=1^{w(\bx_1)}0^{n-w(\bx_1)}$ and $\bx_2=1^{w(\bx_1)-1}0^{n+1-w(\bx_1)}$. Since we are interested in the case, where we require the largest number of channels for distinguishing between two input words, we assume from now on that $\bx_1$ and $\bx_2$ are as in the previous sentence. Furthermore, observe that we may assume without loss of generality that $w(\bx_1)\geq n/2+1$ and denote $w(\bx_1)=w$. Indeed, one of the two words has either more than $n/2$ zeroes or ones and if necessary, we could swap the roles of zeroes and ones. Consider the words $\bx=0^{w-1}10^{n-w}$ and $\bx'=0^{w-2}10^{n+1-w}$. We show that the multiset of output words which can be obtained from $\bx$ and $\bx'$ is at least as large as the multiset of output words which can be obtained from $\bx_1$ and $\bx_2$. Let $D=\{\bd_1,\dots,\bd_N\}$ and $D'=\{\bd_1',\dots,\bd_N'\}$ be the sets of deletion vectors of weight $t$ such that, for each $i\in[1,N]$, if we obtain an output word $\by$ from $\bx_1$ with $\bd_i\in D$, then we also obtain it from $\bx_2$ with $\bd_i'\in D'$. 
Furthermore, we make an observation that if $w(\bx_1)\in \supp(\bd)$ for $\bd\in D$ and  applying $\bd$ to $\bx_1$ gives an output word $\by$, then applying $\bd$ to $\bx_2$ also gives the same output word $\by$. Hence, we may assume that $D$ and $D'$ contain every deletion vector of weight $t$ which have $w(\bx_1)$ in their supports. There are $\binom{n-1}{t-1}$ such deletion vectors. Similarly, for $\bx$ and $\bx'$, we know that deletion vectors which contain $w$ or $w-1$, respectively, in their supports lead to the same output words containing only zeroes and there are $\binom{n-1}{t-1}$ such deletion vectors. Thus, we omit these deletion vectors from all the following considerations and calculations.
Moreover, let us denote by $D_\by\subseteq D$ the set of all deletion vectors of $D$ which result to $\by$ after applying them to $\bx_1$. Similarly, we denote by $D_\by'\subseteq D'$ the set of all deletion vectors of $D'$ which result to $\by$ after applying them to $\bx_2$. Recall that if applying $\bd_i\in D$ to $\bx_1$ leads to $\by$, then applying $\bd_i'\in D'$ to $\bx_2$ leads to $\by$. In particular, we have $|D_\by|=|D_\by'|$ for each $\by$ (since we consider the multiset model). Notice that sets $D_\by$ partition $D$ and sets $D_\by'$ partition $D'$.

We show that for each $D_\by$ (where $\by$ does not belong to the above omitted output words), we can injectively link another output word $\by'$ which can be attained with at least $|D_\by|$ deletion vectors from both $\bx$ and $\bx'$. Let $\by=1^{w-i}0^{n+i-w-t}$ and $\by'=0^{w-1-i}10^{n+i-w-t}$ for $i\in [1,w-1]$ (note that since $w(\by')>0$, it was not omitted above). We may assume that $i\leq w-1$, due to the previous omissions. Clearly, we also have $i\le t$. Let us use following notation:  \begin{align*}
m_1&=\binom{w-1}{i}\binom{n-w}{t-i},\\
m_2&=\binom{w-1}{i-1}\binom{n-w}{t+1-i},\\
m&=\binom{w-1}{i}\binom{n-w}{t-i},\\
m'&=\binom{w-2}{i-1}\binom{n+1-w}{t+1-i}.
\end{align*}We can obtain $\by$ from $\bx_1$ with $m_1$ deletion vectors and from $\bx_2$ with $m_2$ deletion vectors. Moreover, we may obtain $\by'$ from $\bx$ with $m$ deletion vectors and from $\bx'$ with $m'$ deletion vectors.

Notice that $m_1=m$. Next, we consider the values of $i\in[1,\min\{w-1,t\}]$ for which $m'\geq m_2$. Recall that $\binom{a}{b}=0$ if $b<0$ or $b>a$. 
Note that both $m'$ and $m_2$ obtain value $0$ with the same values of $i$, with the possible exception that  $m_2=0$ and $m'\ge1$ when $i=w+t-n$. Further note that, since $i\in[1,w-1]$, the left binomial coefficient in each of four parameters is always positive.
 When both $m'$ and $m_2$ obtain positive values, we have
$$\frac{m'}{m_2}=\frac{(n+1-w)(w-i)}{(w-1)(n+i-w-t)}\geq1 \Leftrightarrow\
%
%
n-t+wt \geq ni \Leftrightarrow
1+\frac{t(w-1)}{n} \geq i.$$

Denote above $P=1+\frac{t(w-1)}{n}$.
Let us next consider when we use $m$ and when $m'$. Recall that for each $i$, we are interested in the one that is smaller. Notice that both $m$ and $m'$ obtain value zero for same values of $i$. Now  for nonzero values
$$\frac{m}{m'}=\frac{(w-1)(t+1-i)}{(n+1-w)i} \geq1
%
%
\Leftrightarrow wt+w-t-1 \geq ni 
\Leftrightarrow \frac{(w-1)(t+1)}{n} \geq i.$$

Denote $Q=\frac{(w-1)(t+1)}{n}$.
Furthermore, let us compare values $m$ (or $m_1$) and $m_2$. Notice that $m$ and $m_2$ obtain value zero for the same values of $i$ with the possible exception for 
$i=t+w-n$ for which we may have $m\geq1$ and $m_2=0$.
Now  for nonzero values
$$\frac{m}{m_2}=\frac{(w-i)(t+1-i)}{i(n+i-w-t)} \geq1
%
%
\Leftrightarrow wt+w \geq ni+i
\Leftrightarrow \frac{w(t+1)}{n+1} \geq i. $$

Denote $ R=\frac{w(t+1)}{n+1}$.
Next, we show that $P\geq R\geq Q$ (since $ w,t\leq n$). We have \begin{align*}
P-R =& 1+\frac{(n+1)t(w-1)-nw(t+1)}{n^2+n}\\
%
%
=&1+\frac{tw-nt-t-nw}{n^2+n}\\
\geq& 1+\frac{tn-nt-t-n^2}{n^2+n}\geq0
\end{align*}
and 
\begin{align*}
R-Q =& \frac{nw(t+1)-(n+1)(w-1)(t+1)}{n^2+n}\\
%
%
%
=& \frac{(n+1-w)(t+1)}{n^2+n}>0.
\end{align*}

We note that for $i\in[1,\min\{t,w-1\}]$ when $m=0$, we also have $m_1,m_2=0$ and also when $m'=0$, we have $m_1,m_2=0$, 
as we can see from above together with the equality $m=m_1$.
Consider next the cases with $i\geq P$ and $P>i\geq R$. In both of these cases $m_2\geq m=m_1$. Now, we  obtain $\by$ and $\by'$ with $m=m_1$ ways since $m'\geq m$ and $m_2\geq m_1$. If $R>i\geq Q$, then $m'\geq m=m_1\geq m_2$ and we obtain $\by$ with $m_2$ deletion vectors and $\by'$ with $m\geq m_2$ deletion vectors. 
Finally, if $i<Q$, then $m=m_1\geq m'\geq m_2$ and we obtain $\by$ with $m_2$ deletion vectors and $\by'$ with $m'\geq m_2$ deletion vectors. 
Thus, in all three cases we can obtain $\by'$ in at least as many ways as we can obtain $\by$. Therefore, for any transmitted word pair $\bx_1,\bx_2$ with difference of exactly one in their weights, there exists another word pair $\bx$ and $\bx'$ with equal weights of one such that we require at least as many channels for distinguishing between $\bx$ and $\bx'$ as we require for distinguishing between $\bx_1$ and $\bx_2$.
%
\end{proof}

\begin{remark}\label{RemComparison} In this remark we discuss wordpairs leading to the largest channel numbers in the different models when exactly $t$ deletion errors occur.

\begin{enumerate}
\item In the traditional model, for even $n$ and $q=2$, the extremal wordpair (up to permutation of symbols) is (see \cite[proof of Lemma 1]{levenshtein2001efficient}) the pair $\bx=01010101\cdots01$, $\bx'=10010101\cdots01$ and for $n=2+hq$, $q\geq2$, $h\in \N$, the extremal wordpair is $\bx=0123\cdots(q-1)012\cdots(q-1)$, $\bx'=1023\cdots(q-1)012\cdots(q-1)$, that is the only difference is in the first two symbols and the words continue afterwards as alternating words.
\item In the deletion pattern model with non-multisets and even $n$, an extremal wordpair is $\bx=0^{n/2}10^{n/2-1}$, $\bx'=0^{n/2-1}10^{n/2}$ by Theorems \ref{thm:delVecnAllnbound}(ii) and \ref{thm:delVecnAllNormalSet}(ii).
\item In the deletion pattern model with multisets, even $n$ and odd $t$, 
 the wordpair $\bx = 0^{n/2}10^{n/2-1}$, $\bx' = 0^{n/2-1}10^{n/2}$, which is given in Corollary \ref{CorDelVecMulExHalf} (up to a permutation of symbols), seems to require the largest number of channels. 
Indeed by the proof of Proposition \ref{DelVecMult1}, it is an extremal wordpair for $t=1$. Furthermore, it is easy to check by computer, using a brute-force method finding every extremal wordpair, that this pair actually belongs to the set of \emph{extremal} wordpairs when $t=3$ and $n\in \{8,10\}$.
\item In the deletion pattern model with multisets, $n=2h(t+1)$ and even $t$, the wordpair which seems to be requiring the largest number of channels, which is presented in Corollary \ref{CorDelVecMulExBad} (up to permutation of symbols), is $\bx=0^{ht}10^{h(t+2)-1}$, $\bx'=0^{ht-1}10^{h(t+2)}$.  It is easy to check by computer with a brute-force method that this pair  belongs to the set of extremal wordpairs when $t=2$ and $n\in \{6,12\}.$ 
\end{enumerate}
\end{remark}


In the subsequent lemma, we give a tool for comparing the number of channels required in the worst case of the non-multiset deletion vector version compared to the multiset version.

\begin{lemma}\label{LembinomCompar}
Let $n\geq2t+2$ and $t$ be even positive integers. We have $$\frac{\binom{n/2-1}{t/2}^2}{\binom{n/2-1}{t}}\overset{n\to\infty}{\longrightarrow}\binom{t}{t/2}.$$
\end{lemma}
\begin{proof}
We have \begin{align*}
	&\frac{\binom{n/2-1}{t/2}^2}{\binom{n/2-1}{t}}
	=\frac{(n/2-1)!t!(n/2-1-t)!}{(n/2-1-t/2)!(n/2-1-t/2)!(t/2)!(t/2)!}\\
	=&\binom{t}{t/2}\frac{(n/2-1)\cdots(n/2-t/2)}{(n/2-1-t/2)\cdots(n/2-t)}\\
	\overset{n\to\infty}{\longrightarrow}&\binom{t}{t/2}.
\end{align*}

\end{proof}

For an even  $t$ when exactly $t$ deletions occur, by Theorems \ref{thm:delVecnAllnbound}(ii) and \ref{thm:delVecnAllNormalSet} we require $\binom{n}{t}-\binom{n/2-1}{t}+1$ channels in the non-multiset model to separate extremal words presented in the case 2) of Remark \ref{RemComparison}. The same wordpair is also mentioned in Remark \ref{RemComparison} for the multiset model with even $t$ and by Corollary \ref{CorDelVecMulExHalf}, we require $\binom{n}{t}-\binom{n/2-1}{t/2}^2+1$ channels to distinguish between these words. By Lemma \ref{LembinomCompar}  when $n$ is large,  we require roughly $$\left(\binom{n}{t}-\binom{n/2-1}{t}+1\right)-\left(\binom{n}{t}-\binom{n/2-1}{t/2}^2+1\right)\approx \left(\binom{t}{t/2}-1\right)\binom{n/2-1}{t}$$
 more channels in the non-multiset model compared to the multiset case.

\section{Decoding}\label{sec:decoding}

In this section, we consider channels with insertion, deletion and substitution errors
using an underlying code containing almost all words of $\F^n_q$. 
We assume that each insertion vector is applied to the word of length $n$. Then deletion vectors and substitutions are applied to original non-inserted  symbols and no deletion affects the substituted symbols. We assume that each error pattern has the same probability. Unlike in the previous section, in this 
section we allow multiple channels to have the same error patterns.  In particular,  if only substitution errors occur, then each possible output word has the same probability to be outputted as we have seen in the beginning of Section \ref{SecErrorPat}. However, in the case of deletion and insertion errors, some output words are more likely. For the rest of the section, we focus on $q\geq4$. Notice that the presented technique cannot be expanded to the cases with $q < 4$ as will be seen in Remark~\ref{Rem:q<4}. Moreover, the case with $q=4$ is a natural size of alphabet for DNA-storage. The case with $q=4$ is presented in the conference version of this article \cite{junnila2023levenshtein} without a proof.

For channels with insertion, deletion and substitution errors, we introduce, for a code with  minor restrictions, a decoding  algorithm with complexity $O(Nn)$, where $N$ is the number of output words read at the point in which the algorithm halts (see Algorithm~\ref{algsubsinsdelLargeq}). Our algorithm never gives an incorrect result. However, for some output sets $Y$ it only outputs an empty word. 
When we discuss about complexities, we assume $q$ to be constant. 
The code we are using has only minor restrictions on how common the two most common symbols in any codeword can be. Moreover, similar restrictions have been used for example in \cite{viswanathan2008improved}. 
Besides giving verifiability properties and solving all three types of errors simultaneously, the novelty of our technique is that we do not use majority decoding which has been an essential part of most earlier techniques. 

Algorithm~\ref{algsubsinsdelLargeq} is an online algorithm in the sense that the output words of the channels can be viewed to be fed to the algorithm one by one (instead of giving all the outputs at once). In this context, the number $N$ of channels is assumed to denote the number of outputs required before the algorithm stops. Moreover, the algorithm is sort of a randomized one in style of a Las Vegas algorithm, although technically the randomization occurs outside of the algorithm in obtaining the output words of the channels. However, in Las Vegas style, if the number of the output words is unrestricted, then the algorithm is not guaranteed to halt (although it is highly likely), but if the algorithm halts, then it always gives a correct result.

Probabilistic decoders have been previously mostly considered for a setup, where each error to a single coordinate has an independent chance to occur, under the name \textit{trace reconstruction}; see, for example, \cite{batu2004reconstructing} in the case of deletion channels and \cite{viswanathan2008improved} in the case of simultaneous insertion, deletion and substitution errors. Unlike in these setups, we limit the maximum number of errors   which may occur in a channel, as has been done, for example, by Levenshtein in \cite{levenshtein2001efficient}. That allows our algorithm to have verifiability, that is, although the algorithm is probabilistic, it is likely that the algorithm halts (see Lemma \ref{ProbabilityALGDelInsSubLargeq}), and the output is always correct if the algorithm halts (see Lemma \ref{CorrectnessALGDelInsSubLargeq}). 

Let code $C \subseteq \Z_q^n$ contain all the words of $\Z_q^n$ except for those in which the two most common symbols appear together in total in at least $\lceil(p-1)n/p\rceil$ positions with $p=2^4/\mathbbm{e}$. Observe that there are $$\binom{q}{2}\left(\sum^n_{i=\left\lceil\frac{(p-1)n}{p}\right\rceil}\binom{n}{i}2^i(q-2)^{n-i}\right)$$ such words. Due to these restrictions on $C$, the third most common symbol (and also second most common symbol) in any codeword occurs in at least $\lceil n/((q-2)p)\rceil$ coordinates by the pigeonhole principle. Notice that $p$ is irrational and hence, $(p-1)n/p$ is not an integer. Our results in this section require that we are using code $C$ (or some sub-code of $C$). We next show that $C$ is \textit{large} when $q$ is fixed and $n$ is large.  
In order to estimate the cardinality of $C$, we first consider the case with $q = 4$. We have
\begin{align}
|C|\geq& 4^n-\binom{4}{2}\left(\sum^n_{i=\left\lceil\frac{(p-1)n}{p}\right\rceil}\binom{n}{i}2^i(4-2)^{n-i}\right)\nonumber\\
=&4^n-6\cdot4^{n/2}\sum^n_{i=\left\lceil\frac{(p-1)n}{p}\right\rceil}\binom{n}{i}\nonumber\\
\label{Csize}\tag{*}\geq& 4^{n}-6\cdot4^{n/2}(\mathbbm{e}p)^{n/p}\\
=& 4^{n}-6\cdot4^{n/2}4^{2\mathbbm{e}n/2^4}\nonumber\\
>&4^{n}-6\cdot 4^{7n/8}\in\Theta(q^n)\label{Csizeq=4}.
\end{align}

In Inequality (\ref{Csize}), we use the following modification of a well-known upper bound for partial binomial sums: $$\sum_{i=0}^{\lfloor h\rfloor}\binom{K}{i}\leq \left(\frac{\mathbbm{e}n}{h}\right)^h,$$ where $K\in \Z$ and $K\geq h>0$. Indeed, this upper bound holds since

\begin{align*}\sum_{i=0}^{\lfloor h\rfloor}\binom{K}{i}&\leq \sum_{i=0}^{\lfloor h\rfloor}\frac{h^i}{i!}\cdot\left(\frac{K}{h}\right)^i\leq \left(\frac{K}{h}\right)^{\lfloor h\rfloor}\sum_{i=0}^{\lfloor h\rfloor}\frac{h^i}{i!}\\
&< \left(\frac{K}{h}\right)^{\lfloor h\rfloor}\mathbbm{e}^h\leq \left(\frac{\mathbbm{e}K}{h}\right)^{ h}.
\end{align*}
In particular, for $\sum^n_{i=\left\lceil\frac{(p-1)n}{p}\right\rceil}\binom{n}{i}$ it gives:  $$\sum^n_{i=\left\lceil\frac{(p-1)n}{p}\right\rceil}\binom{n}{i}=\sum^{\lfloor n/p\rfloor}_{i=0}\binom{n}{i}\leq \left(\frac{\mathbbm{e}n}{n/p}\right)^{n/p}=(\mathbbm{e}p)^{n/p}.$$ 

 We can use similar arguments for the case with $q\geq 5$. Let $q=2^b\geq5$, $b=b'+\log_2 5$ where $b'\geq0$. We have \begin{align}
|C|=& q^n-\binom{q}{2}\left(\sum^n_{i=\left\lceil\frac{(p-1)n}{p}\right\rceil}\binom{n}{i}2^i(q-2)^{n-i}\right)\nonumber\\
>&q^{n}-\frac{q^2}{2}\cdot2^{n}\cdot q^{n/p}\sum^n_{i=\left\lceil\frac{(p-1)n}{p}\right\rceil}\binom{n}{i}\nonumber\\
>& q^{n}-q^2\cdot2^{n}\cdot 2^{bn/p}(\mathbbm{e}p)^{n/p}\nonumber\\
\geq& q^{n}-2^{2b}\cdot2^{n}\cdot 2^{\mathbbm{e}bn/16}\cdot2^{4\mathbbm{e}n/16}\nonumber\\
=&q^{n}-2^{2b}\cdot2^{n+\mathbbm{e}\log_2 (5)n/16+4\mathbbm{e}n/16+\mathbbm{e}b'n/16}\nonumber\\
%
>&q^{n}-2^{2b}\cdot2^{2.08n+b'n/5}
\in\Theta(q^n).\label{CsizeLargeqLastLine}
\end{align}
The inclusion (\ref{CsizeLargeqLastLine}) follows from $q^n=2^{bn}=2^{\log_25n+b'n}>2^{2.32n+b'n}$ and the facts that $2.32>2.08$ and $b'>b'/5$.
By (\ref{Csizeq=4}) and (\ref{CsizeLargeqLastLine}), we have  $|C|\in \Theta(q^n)$ for all integers $q\geq 4$.

 We denote by $t_s$, $t_i$ and $t_d$ the number of substitution, insertion and deletion errors, respectively, which may occur in a channel. When we discuss about the complexity of our algorithm, these values are assumed to be constants. Moreover, Lemma \ref{ProbabilityALGDelInsSubLargeq} gives them some minor constraints. 
Recall that for our Las Vegas algorithm, the underlying code $C \subset \Z_q^n$ is required to be such that 
in each codeword the two most common symbols appear in total in at most $(p-1)n/p\approx 0.83n$ positions. When $p\geq 2^4/\mathbbm{e}\approx 5.9$, we have $|C|\in \Theta(q^n)$.

\begin{remark}\label{RempEffect} 
 Observe that if we increase the value of $p$ from $2^4/\mathbbm{e}$, then that will increase the size of the code $C$. However, we have a trade-off later in the proof of Lemma \ref{ProbabilityALGDelInsSubLargeq}; the larger $p$ is the less likely Algorithm \ref{algsubsinsdelLargeq} is to stop.
 \end{remark}
%

We denote $t_m=t_d+t_i+2t_s$ and for a word $\bw=(w_1,w_2,\dots,w_n)$ 
we denote $$M_i(\bw)=|\{j\mid w_j=i\in\Z_q\}|$$ and $$M_{a,b,c}(\bw)=|\{j\mid w_j\not\in\{a,b,c\}\}|.$$ The useful observation behind  Algorithm \ref{algsubsinsdelLargeq} is that $M_i(\by)+t_m\geq M_i(\by')$ for every $i$ and any two output words $\by, \by' \in Y$ and this bound can be attained when  $M_i(\by)\geq t_d+t_s$. This observation is further discussed in the proof of the following lemma.

\begin{lemma}\label{tm errorsLargeq}
Let $a$, $b$ and $c$ be distinct symbols of $\Z_q$. 
\begin{enumerate}
\item If $\by_1,\by_2\in Y$ are such that $M_a(\by_1)=M_a(\by_2)+t_m$, then $\by_1$ is formed from the transmitted word $\bx$ by inserting $t_i$ symbols $a$ and substituting $t_s$ symbols  by  $a$, and $\by_2$ is formed from $\bx$ by deleting $t_d$   symbols $a$ and substituting $t_s$ symbols $a$ with other symbols.
\item If $\by_1,\by_2,\by_3\in Y$ are such that $M_a(\by_1)=M_a(\by_2)+t_m$ and $M_b(\by_3)=M_b(\by_1)+t_m$, then $\by_1$ is formed from $\bx$ by inserting $t_i$ symbols $a$, substituting $t_s$ symbols $b$ by  $a$ and deleting $t_d$ symbols $b$.
\item If $\by_1,\by_2\in Y$ are such that $M_a(\by_1)=M_a(\by_2)+t_m$ and $M_{a,b,c}(\by_2)=M_{a,b,c}(\by_1)+t_i+t_s$, then $\by_2$ is formed from $\bx$ by inserting $t_i$ symbols other than $a$, $b$ or $c$, substituting $t_s$ symbols $a$ by symbols other than $a$, $b$ or $c$ and deleting $t_d$ symbols $a$.
\end{enumerate}
\end{lemma}
\begin{proof}
Recall that $t_m=t_i+t_d+2t_s$. Let us first prove Claim $1.$ Observe that we have $M_a(\by_1)\leq M_a(\bx)+t_i+t_s$ since only insertions and substitutions may increase the number of symbols $a$ in an output word and that the equality holds only when all insertions and substitutions increase the number of symbols $a$. Furthermore, we have $M_a(\by_2)\geq M_a(\bx)-t_s-t_d$ since only deletions and substitutions may decrease the number of symbols $a$ in an output word, and the equality holds only when all deletions and substitutions decrease the number of symbols $a$. Thus, $M_a(\by_1)\leq M_a(\by_2)+t_m$ and the equality holds only when  all insertions and substitutions increase the number of symbols $a$ in $\by_1$ and  all deletions and substitutions decrease the number of symbols $a$ in $\by_2$. Hence, the Claim follows.

Claim $2.$ is a direct corollary from Claim $1$.

Claim $3.$ follows from Claim $1.$ Indeed, by Claim $1$, each symbol deleted or substituted out of $\bx$ to form $\by_2$ is $a$. Moreover, word $\by_2$ has $t_i+t_s$ more symbols in total in the set $\F_q\setminus\{a,b,c\}$ than word $\by_1$. Thus, each symbol which we insert or substitute to $\bx$ to form $\by_2$ is in $\F_q\setminus\{a,b,c\}$.
%
\end{proof}

	\begin{algorithm}[th!]
	\caption{Decoding  in $\F_q^n$}\label{algsubsinsdelLargeq}
	\begin{algorithmic}[1]
	\Require  At least six output words $Y=\{\by_i\mid i\in \Z_+\}$
	\Ensure Transmitted word $\bc(=\bx)$ or empty word $\varepsilon$
\item Let $i=1$, $\bc=\varepsilon$, collection $Y_{2}=\{\by_{a,b}\mid a\neq b, a,b\in\F_q\}$, collection $Y_{3}=\{\by_{a,b,c}\mid a\neq b\neq c\neq a, a,b,c\in\F_q\}$, $Y_6=\emptyset$ and $\by_{a,b,c}=\by_{a,b}=\by_1$ for each $a\neq b\neq c\neq a$
\While{ $Y_6 = \emptyset$ and $i \leq |Y|$}
\State Read $\by_i\in Y$ 
\For{each $j, j',j''\in[0,q-1]$ with $j\neq j'\neq j''\neq j$} 
\State calculate $M_j(\by_i)$ and $M_{j,j',j''}(\by_i)$
\EndFor
\For{each $\by_{a,b}\in Y_{2}$}
\If{ $M_a(\by_i)\leq M_a(\by_{a,b})$ and $M_b(\by_i)\geq M_b(\by_{a,b})$}
\State Set $\by_{a,b}:=\by_{i}$ and store $M_a(\by_i)$ and $M_b(\by_i)$
\EndIf
\EndFor
\For{each $\by_{a,b,c}\in Y_{3}$}
\If{ $M_a(\by_i)\leq M_a(\by_{a,b,c})$ and $M_{a,b,c}(\by_i)\geq M_{a,b,c}(\by_{a,b,c})$}
\State Set $\by_{a,b,c}:=\by_{i}$ and store $M_a(\by_i)$ and $M_{a,b,c}(\by_i)$
\EndIf
\EndFor
\If{  there exist in $Y_{2}$ words $\by_1=\by_{i_3,i_1}$, $\by_2=\by_{i_1,i_2}$, $\by_3=\by_{i_2,i_3}$ and in $Y_3$ words $\by_4=\by_{i_1,i_2,i_3}$, $\by_5=\by_{i_2,i_1,i_3}$ and $\by_6=\by_{i_3,i_1,i_2}$ such that $i_j\neq i_h$ for all distinct $j,h$ as well as
\begin{align*}
M_{i_1}(\by_1)&=M_{i_1}(\by_2)+t_m,\hspace{0.1cm} M_{i_2}(\by_2)=M_{i_2}(\by_3)+t_m,  \\
M_{i_3}(\by_3)&=M_{i_3}(\by_1)+t_m,\hspace{0.1cm} M_{i_1}(\by_2)=M_{i_1}(\by_4),\\
M_{i_2}(\by_3)&=M_{i_2}(\by_5),\hspace{0.925cm}M_{i_3}(\by_1)=M_{i_3}(\by_6)\text{ and}\\
M_{i_1,i_2,i_3}(\by_4)&=M_{i_1,i_2,i_3}(\by_5)=M_{i_1,i_2,i_3}(\by_6)=M_{i_1,i_2,i_3}(\by_1)+t_i+t_s 
\end{align*}
\State}
\State $ $ Set $Y_{6}=\{\by_j\mid j\in[1,6]\}$
\EndIf
\State Set $i=i+1$
\EndWhile
\If{$Y_6=\emptyset$}
\State\Return empty word
\EndIf

\item Delete each $i_1$ and $i_3$ from $\by_1(=\by_{i_3,i_1})$ to construct $\bz_1$


\item Delete each $i_1$ and $i_2$ from $\by_2(=\by_{i_1,i_2})$ to construct $\bz_2$


\item Delete each $i_2$ and $i_3$ from $\by_3(=\by_{i_2,i_3})$    to construct $\bz_3$


\item Delete everything except each $i_2$ and $i_3$ from $\by_4(=\by_{i_1,i_2,i_3})$     to construct $\bz_4$


\item Delete everything except each $i_1$ and $i_3$ from $\by_5(=\by_{i_2,i_1,i_3})$    to construct $\bz_5$


\item Delete  everything except each $i_1$ and $i_2$ from $\by_6(=\by_{i_3,i_1,i_2})$    to construct $\bz_6$


\While{there exists an index $j$ such that $\bz_j\neq\mathbf{\varepsilon}$}
\If{exactly three different words $\bz_i,\bz_j$ and $\bz_h$ start with the same symbol $a$}
\State Concatenate $\bc$ from right with $a$
\State Remove the first symbol of $\bz_i,\bz_j$ and $\bz_h$
\EndIf
\EndWhile
\State\Return $\bc$
			\end{algorithmic}
	\end{algorithm}

In the following lemmas, we first show that Algorithm \ref{algsubsinsdelLargeq} never gives an incorrect output and that it  is efficient. 
Then we show that we are likely to find the  set $Y_6$.

\begin{lemma}\label{CorrectnessALGDelInsSubLargeq}
If, after reading exactly $N$ inputs, we find output words $\by_i\in Y_6$, $i\in[1,6]$, defined in Step $17$ of Algorithm \ref{algsubsinsdelLargeq}, then $\bc=\bx$ in Algorithm \ref{algsubsinsdelLargeq} and algorithm halts in $O(Nn)$ time. 
\end{lemma}
\begin{proof}
Let $\by_i$ and $\bz_i$, $i\in[1,6]$, be as in Algorithm \ref{algsubsinsdelLargeq}. Consider first the output words $\by_1$, $\by_2$ and $\by_3$. Observe that $M_{i_1}(\by_1) = M_{i_1}(\by_2)+t_m$ and $M_{i_3}(\by_3) = M_{i_3}(\by_1)+t_m$. Therefore, by Lemma~\ref{tm errorsLargeq}, each of $t_i$ inserted symbols in $\by_1$ is $i_1$, each of $t_d$ deleted symbols is $i_3$ and all $t_s$ substitutions change symbols $i_3$ to symbols $i_1$. Consequently, the output word $\by_1$ is obtained from the (unknown) transmitted word $\bx$ by modifying only the symbols $i_1$ and $i_3$.  Similarly   (due to the three first equations in Step $17$)
modifications to $\bx$ in obtaining $\by_2$ affect only the symbols  $i_1$ and $i_2$ and
modifications to $\bx$ in obtaining $\by_3$ affect only the symbols  $i_2$ and $i_3$. Observe that at this point we know the exact number of each symbol in $\bx$  (but not their order). In particular, $$M_{i_1}(\bx)=M_{i_1}(\by_3),M_{i_2}(\bx)=M_{i_2}(\by_1),M_{i_3}(\bx)=M_{i_3}(\by_2)$$ and $$M_{i_4}(\bx)=M_{i_4}(\by_1)=M_{i_4}(\by_2)=M_{i_4}(\by_3)$$ for any $i_4\not\in\{i_1,i_2,i_3\}.$ Let us then consider the output words $\by_4,\by_5$ and $\by_6$.

By the previous observations, we first obtain $M_{i_1}(\by_1) = M_{i_1}(\by_2) + t_m = M_{i_1}(\by_4) + t_m$. Therefore, as $M_{i_1,i_2,i_3}(\by_4) = M_{i_1,i_2,i_3}(\by_1) + t_i + t_s$, we obtain by Lemma~16(3) that $\by_4$ is formed from $\bx$ by adding $t_i + t_s$ symbols (with insertions or substitutions) other than $i_1$, $i_2$ or $i_3$ and by removing $t_d + t_s$ symbols $i_1$ (with deletions or substitutions). Similarly, we obtain that the symbols added to $\by_5$ and $\by_6$ are other that $i_1$, $i_2$ or $i_3$ and the removed symbols are $i_2$ and $i_3$, respectively.

Consequently, if we consider the four symbol types  examined above, namely $i_1$, $i_2$, $i_3$ and $\F_q\setminus\{i_1,i_2,i_3\}$. The modifications within each word $\by_i$, $i\in[1,6]$, with respect to $\bx$ are restricted to symbols in two of the examined types.  Thus, we know that symbols  in $\bz_i$ ($i=[1,6]$) are ordered in the same way as in the transmitted word $\bx$, since we have removed all modified symbols from $\by_i$ when we have formed $\bz_i$. Furthermore, we have $\binom{4}{2}=6$ different words $\bz_i$ and for each pair of the missing symbol types, we have a word $\bz_i$ from which exactly those types are missing.

Next we show that we obtain the transmitted codeword $\bc=\bx$ in Algorithm \ref{algsubsinsdelLargeq} during Steps 32--37. If, for example, $x_1= i_1$, then the first symbol of $\bz_{3},\bz_{5}$ and $\bz_{6}$ is $x_1$. Moreover, $\bz_1, \bz_2$ and $\bz_4$ cannot share a common first symbol. The same is true for $x_1= i_j$ for any $i_j\in\{i_1,i_2,i_3\}$ since words $\bz_i$ go through all $\binom{4}{2}=6$ combinations of missing symbol type pairs among the four examined symbol types. Furthermore, if $x_1\in\F_q\setminus\{i_1,i_2,i_3\}$, then 
$x_1$ is equal to the first symbol of $z_{1}$, $z_{2}$ and $z_{3}$. 
Therefore, in all cases, we have $c_1=x_1$. As we go on, we remove the first symbol from those $\bz_i$'s which shared the same symbol. By iteratively applying these arguments, we obtain the rest of the symbols of $\bx$.

Let us then consider the complexity of the algorithm. Here, we assume that $q$ is a constant on $n$. We observe that in the first while loop between Steps $2$ and $22$, we only do simple coordinatewise comparison operations and the loop lasts at most $N$ rounds. Between Steps $26$ and $31$, we again make only simple modifications to the words of length $n+t_i-t_d$. Finally, all operations in the final while loop occur to words of length at most $n-t_d$ and the operations are  simple. Hence, the complexity of the algorithm is in $O(Nn)$.
\end{proof}

\begin{lemma}\label{ProbabilityALGDelInsSubLargeq}
As $N$ increases, the probability for obtaining output words $\by_i\in Y_6$, $i\in [1,6]$, in Step $17$ of Algorithm \ref{algsubsinsdelLargeq} approaches $1$ for any $n\geq (q-1)p(t_d+t_s)$. 
\end{lemma}
\begin{proof}
Consider the set $Y_6$ in Algorithm \ref{algsubsinsdelLargeq}. Recall from the proof of Lemma \ref{CorrectnessALGDelInsSubLargeq} the separation of symbols into four types $i_1,i_2,i_3$ and $\F_q\setminus\{i_1,i_2,i_3\}$.
 Moreover, we see (as in Lemma \ref{CorrectnessALGDelInsSubLargeq}) from the equations in Step $17$ of Algorithm \ref{algsubsinsdelLargeq}  that $Y_6$ has six words and each of them can be obtained from $\bx$ by modifying the symbols of exactly two symbol types. In particular, we observe that 
 for each symbol pair $i_j$, $i_h$ ($j\neq h$ and $j,h\in \{1,2,3\}$) there exists a word $\by_i\in Y_6$ which is formed from $\bx$ by focusing all the modifications to these two symbols. Moreover, the symbols of $\Z_q\setminus \{i_1,i_2,i_3\}$ 
 are such that they are never removed from $\bx$ to form these output words. 
Moreover, there are multiple possible ways  (regarding the symbols) in which we can form the subset $Y_6$ from $Y_{2}$ and $Y_3$ and it is enough for our claim that we find at least one of these ways. Furthermore, if a set of words in $Y$ satisfies the conditions  set for $Y_6$ in Step 17, then those words are found in Steps 7 to 16. Let us assume without loss of generality that $i_1,i_2$ and $i_3$ are the three most common symbols in $\bx\in C\subseteq \F^n_q$ and $M_{i_3}(\bx)\leq M_{i_2}(\bx)\leq M_{i_1}(\bx)$.  Recall, that our restrictions on code $C$ guarantee, that $M_{i_3}(\bx)\geq \lceil n/((q-2)p)\rceil$ by the pigeonhole principle.

 Thus, here we consider only the case where we remove symbols $i_1$, $i_2$ and $i_3$. 
 Notice that the likelihood of obtaining exactly this kind set $Y_6$ is less than the likelihood of obtaining any suitable set $Y_6$. 
Now, the least likely case is the one where we remove symbols $i_3$ from $\bx$ since $i_3$ is the least common among $\{i_1,i_2,i_3\}$. We denote that word by $\by_1$ and the symbol we insert to it is assumed to be $i_1$ (all symbols have equal probability to be inserted).
 Notice that since $n\geq (q-1)p(t_d+t_s)$, we have $M_{i_3}(\bx)\geq \lceil n/((q-2)p)\rceil\geq t_d+t_s$.  

In the subsequent approximations, we will need the following well-known lower bound. If $K,h$ be such non-negative integers that $K\geq3h-1$, then we have \begin{align}\label{2binom>ballLargeq}
2\binom{K}{h}=&\binom{K}{h}+\frac{K!}{h!(K-h)!}\nonumber\\
=&\binom{K}{h}+\frac{K!}{(h-1)!(K-h+1)!}\cdot\frac{K-h+1}{h}\nonumber\\
\geq&\binom{K}{h}+2\binom{K}{h-1}\\
\geq&\binom{K}{h}+\binom{K}{h-1}+2\binom{K}{h-2}\nonumber\\
\geq&\cdots\geq V_2(K,h).\nonumber
\end{align}

Let us first consider the probability to obtain the word $\by_1$. To obtain it, $t_i$ insertions occur and each insertion contains only symbol $i_1$. Recall that the likelihood of any specific insertion is $1/|B^I_{t_i}(\bx)|$. First the probability that exactly $t_i$ (for a positive $t_i$) insertions occur is at least 
$\frac{1}{t_i+1}.$ Indeed, by Equation (\ref{insBall}) we have $$\frac{|B_{t_i}^I(\bx)|-|B_{t_i-1}^I(\bx)|}{|B_{t_i}^I(\bx)|}\geq\frac{q^{t_i}\binom{n+t_i}{t_i}}{(t_i+1)q^{t_i}\binom{n+t_i}{t_i}}=\frac{1}{t_i+1}.$$

Probability that each newly inserted symbol is $i_1$ is $\left(\frac{1}{q}\right)^{t_i}.$

Next, we give a lower bound for the probability that each deletion and substitution modifies symbol $i_3$ and that there occurs exactly $t_s$ substitutions and exactly $t_d$ deletions. We assume here that we cannot substitute and delete the same symbol or any inserted symbol. In particular, there are at least $(q-1)^{t_s}\binom{\lceil n/((q-2)p)\rceil}{t_s+t_d}\binom{t_s+t_d}{t_s}$ ways in which the $t_s+t_d$ deletions and substitutions may occur. Moreover, we may apply $i\leq t_s$ substitutions and $j\leq t_d$ deletions to $\bx$ in $(q-1)^i\binom{n}{i+j}\binom{i+j}{i}$ different ways. Hence, for the lower bound of the considered probability, we have 

\begin{align}
&\frac{(q-1)^{t_s}\binom{\lceil n/((q-2)p)\rceil}{t_s+t_d}\binom{t_s+t_d}{t_s}}{\sum_{j=0}^{t_d}\sum_{i=0}^{t_s}(q-1)^i\binom{n}{i+j}\binom{i+j}{i}}\nonumber\\
\geq&\frac{\binom{\lceil n/((q-2)p)\rceil}{t_s+t_d}\binom{t_s+t_d}{t_s}}{\sum_{j=0}^{t_d}\sum_{i=0}^{t_s}\binom{n}{i+j}\binom{i+t_d}{i}}\nonumber\\
\geq&\frac{\binom{\lceil n/((q-2)p)\rceil}{t_s+t_d}}{\sum_{j=0}^{t_d}\sum_{i=0}^{t_s}\binom{n}{i+j}}\nonumber\\
\geq&\frac{\binom{\lceil n/((q-2)p)\rceil}{t_s+t_d}}{\sum_{j=0}^{t_d}V_2(n,j+t_s)}\nonumber\\
\geq&\frac{\binom{\lceil n/((q-2)p)\rceil}{t_s+t_d}}{2\sum_{j=0}^{t_d}\binom{n}{j+t_s}}\label{ballbinom1Largeq}\\
\geq&\frac{\binom{\lceil n/((q-2)p)\rceil}{t_s+t_d}}{2V_2(n,t_d+t_s)}\nonumber\\
\geq&\frac{\binom{\lceil n/((q-2)p)\rceil}{t_s+t_d}}{4\binom{n}{t_d+t_s}}\label{ballbinom2Largeq}\\
=&\frac{\lceil n/((q-2)p)\rceil!(n-t_d-t_s)!}{4(\lceil n/((q-2)p)\rceil-t_d-t_s)!n!}\nonumber\\
\geq&\frac{1}{4}\cdot\left(\frac{\lceil n/((q-2)p)\rceil+1-t_d-t_s}{n}\right)^{t_d+t_s}\nonumber\\
\geq&\frac{1}{4}\cdot\left(\frac{1}{(q-2)p}-\frac{t_d+t_s}{n}\right)^{t_d+t_s}
\geq\frac{1}{4}\cdot\left(\frac{1}{(q-2)(q-1)p}\right)^{t_d+t_s}.\nonumber
\end{align}
Inequalities (\ref{ballbinom1Largeq}) and (\ref{ballbinom2Largeq}) are due to Inequality (\ref{2binom>ballLargeq}). Observe that the condition $K\geq 3h-1$ in Inequality (\ref{2binom>ballLargeq}) is satisfied since $n> 3(t_d+t_s)$.

Finally, the probability that each substitution produces $i_1$ is  $\left(\frac{1}{q-1}\right)^{t_s}.$

Observe that each of these probabilities is positive and can be bounded from below by a positive constant $$A\geq \left(\frac{1}{q-1}\right)^{t_s}\cdot \frac{1}{4}\left(\frac{1}{(q-2)(q-1)p}\right)^{t_d+t_s} \cdot \left(\frac{1}{q}\right)^{t_i} \cdot \frac{1}{t_i+1} $$ which does not depend on $n$. Hence, the probability for not obtaining $\by_1$ in a channel is at most $(1-A)^N$ which tends to $0$ as $N$ grows. Furthermore, we are less or equally likely to obtain $\by_1$ than $\by_i$ for other values of $i$ since $ M_{i_3}(\bx)\leq M_{i_2}(\bx)\leq M_{i_1}(\bx)$. Note that for $\by_4,\by_5$ and $\by_6$ we may have more options (depending on whether $q\geq5$) for symbols which we can insert or substitute into these words and hence, the probability to obtain these words is at least the same as the probability to obtain $\by_1$. Thus, the probability to obtain the output words in $Y_6$ tends to $1$ as $N$ grows.
\end{proof}

In the following example, we consider how Algorithm \ref{algsubsinsdelLargeq} works after we have obtained output words in $Y_6$. 

\begin{example}
Consider the transmitted word $\bx\in \F_6^{10}$ in Table \ref{ExampleTable} together with $t_i=2,t_d=t_s=1$, output set $Y_6$ and words $\bz_i$. We have presented words $\by_j\in Y_6$ in the table.  Notice that values $q,t_d,t_s$ and $n$ do not satisfy condition $n\geq(q-1)p(t_d+t_s)$ of Lemma \ref{ProbabilityALGDelInsSubLargeq}. However, this is not a problem since the requirement was established only for making sure that we obtain set $Y_6$ with high probability and hence, we do not have to worry about Lemma \ref{ProbabilityALGDelInsSubLargeq}.

Let us now consider Steps from 26 to 31 of the algorithm. 
\begin{enumerate}
\item $c_1=1$ and the first bits of $\bz_1,\bz_4$ and $\bz_6$ are deleted. 

\item $c_2=2$ and the first bits of $\bz_j$  are deleted ($j\in\{2,4,5\}$). 

\item $c_3=0$  and the first bits of  $\bz_j$ are deleted ($j\in\{3,5,6\}$).  We continue iterating the process in this way.

\item   $c_4=0$ and  the first bits of  $\bz_j$ are deleted ($j\in\{3,5,6\}$). 

\item    $c_5=3$ and  the first bits of  $\bz_j$ are deleted ($j\in\{1,2,3\}$). At this point, we have $\bz_1=11$, $\bz_2=22$, $\bz_3=0$, $\bz_4=2121$, $\bz_5=202$ and $\bz_6=101$. 

\item   $c_6=2$ and the first bits of  $\bz_j$ are deleted ($j\in\{2,4,5\}$). 

\item   $c_7=1$ and the first bits of  $\bz_j$ are deleted ($j\in\{1,4,6\}$). 

\item[8.]  $c_8=0$ and the first bits of  $\bz_j$ are deleted ($j\in\{3,5,6\}$). Word $\bz_3$ becomes empty but the algorithm continues.

\item[9.] $c_9=2$ and the first bits of  $\bz_j$ are deleted ($j\in\{2,4,5\}$). 

\item[10.] Finally, we get $c_{10}=1$. Now,  $\bc=\bx$ as claimed.

\end{enumerate}\end{example}
\begin{table}[t]\caption{Word $\bx$, set $Y_6$ and words $\bz_i$.}\label{ExampleTable}\centering
\begin{tabular}{|l|l|l|l|l|l|l|l|l|l|l|l|}
\hline
$\bx$ & 1 &  2&  0&  0&  3& 2 & 1 & 0 & 2 & 1 &  \\ \hline
$\by_1(=\by_{2,0})$ & 1 & 0 & 0 & 0 & 3 & 0 & 1 & 0 & 2 & 1 & 0 \\ \hline
$\by_2(=\by_{0,1})$ & 1 & 2 & 1 & 3 & 2 & 1 & 1 & 0 & 1 & 2 & 1 \\ \hline
$\by_3(=\by_{1,2})$ & 2 & 2 & 0 & 0 & 3 & 2 & 0 & 2 & 2 & 1 & 2 \\ \hline
$\by_4(=\by_{0,1,2})$ & 3 & 1& 2 & 0 & 3 & 2 & 1 & 3 & 2 & 4 & 1   \\ \hline
$\by_5(=\by_{1,0,2})$ & 3 & 4 & 2 & 0 & 3 & 0 & 3 & 2 & 0 & 2 & 1 \\ \hline
$\by_6(=\by_{2,0,1})$ & 3 & 1 & 0 & 0 & 3 & 3 & 5 & 1 & 0 & 2 & 1 \\ \hline
$\bz_1$ & 1 & 3 & 1 & 1 &  &  &  &  &  &  &  \\ \hline
$\bz_2$ & 2 & 3 & 2 & 2 &  &  &  &  &  &  &  \\ \hline
$\bz_3$ & 0 & 0 & 3 & 0 &  &  &  &  &  &  &  \\ \hline
$\bz_4$ & 1 & 2 & 2 & 1 & 2 & 1 &  &  &  &  &  \\ \hline
$\bz_5$ & 2 & 0 & 0 & 2 & 0 & 2 &  &  &  &  &  \\ \hline
$\bz_6$ & 1 & 0 & 0 & 1 & 0 & 1 &  &  &  &  &  \\ \hline
\end{tabular}
\end{table}	
	

\begin{remark}\label{Rem:q<4}
Algorithm \ref{algsubsinsdelLargeq} requires that $q\geq4$. Let us consider the case with $q=3$. If the insertion, deletion and substitution errors occur in some word $\by$, for example, to symbols $0$ and $1$, then we only know how many symbols $2$ there are in $\bx$ but we do not know their location in respect to other symbols. This prevents us from reconstructing the transmitted word $\bx$ in a similar way.
\end{remark}
	
Recall Lemma~\ref{ProbabilityALGDelInsSubLargeq} in which we showed that the probability of finding a suitable set $Y_6$ of output words in the algorithm approaches $1$ as $N$ increases. In addition to the asymptotical result of the lemma, we have also run some simulations for obtaining estimates on the exact number of required channels when $q=4$. The simulations have been performed in a rather simple and straightforward manner: The given number of (at most) $t_s$ substitution, $t_d$ deletion and $t_i$ insertion errors have been randomly applied to an arbitrarily chosen transmitted word $\bx \in C$ and then channel outputs have been read until the set $Y_6$ has been obtained. In Table~\ref{Table_ids_algorith_simulations}, for chosen lengths $n$ and number of different errors, we have given an average and median number of channels required when the simulations have run for $100000$ samples. It should be noted that in each case the number of $100000$ samples seems to be enough for the average and median values to converge to the extent that they give a sensible approximation on the number of required channels.

\begin{table}
	\centering
	\caption{The simulations with $100000$ samples for approximating the average and median number of required channels for $q=4$ and various choices of $n$, $t_s$, $t_d$ and $t_i$.} \label{Table_ids_algorith_simulations}
	\begin{tabular}{|c|c|c|c|c|c|}
		\hline
		$n$ & $t_s$ & $t_d$ & $t_i$ & Average & Median \\ \hline
		20  &   1   &   1   &   1   &   489   &  390   \\ \hline
		60  &   1   &   1   &   1   &   310   &  280   \\ \hline
		100 &   1   &   1   &   1   &   288   &  263   \\ \hline
		200 &   1   &   1   &   1   &   274   &  252   \\ \hline
		100 &   2   &   1   &   1   &  3940   &  3506  \\ \hline
		100 &   1   &   2   &   1   &  1310   &  1166  \\ \hline
		100 &   1   &   1   &   2   &  1163   &  1059  \\ \hline
		100 &   1   &   2   &   2   &  5243   &  4685  \\ \hline
		100 &   0   &   0   &   1   &    7    &   6    \\ \hline
		100 &   0   &   1   &   1   &   21    &   20   \\ \hline
		100 &   0   &   0   &   2   &   32    &   29   \\ \hline
		100 &   0   &   0   &   3   &   133   &  118   \\ \hline
	\end{tabular}    
\end{table}

Based on Table~\ref{Table_ids_algorith_simulations}, we can make the following observations which also seem plausible by the analytical study of the algorithm:
\begin{itemize}
	\item The number of required channels decreases when the length $n$ increases.
	\item The substitution errors are the most difficult ones for the algorithm to handle.
	\item The algorithm works surprisingly well when no substitution errors occur.
\end{itemize}

\bibliographystyle{IEEEtran}
\bibliography{ISIT2019}

\end{document}